\documentclass[journal]{IEEEtran}
\usepackage{amsmath,amsfonts}
\usepackage{algorithmic}
\usepackage{algorithm}
\usepackage{array}
\usepackage[caption=false,font=normalsize,labelfont=sf,textfont=sf]{subfig}
\usepackage{textcomp}
\usepackage{stfloats}
\usepackage{url}
\usepackage{verbatim}
\usepackage{graphicx}
\usepackage{cite}

\usepackage{amsthm}
\usepackage{amsmath,bm}
\usepackage{amstext,amssymb}
\usepackage{makecell}

\newtheorem{definition}{Definition}
\newtheorem{theorem}{Theorem}
\newtheorem{example}{Example}
\newtheorem{remark}{Remark}
\newtheorem{lemma}{Lemma}
\newtheorem{proposition}{Proposition}

\begin{document}

\title{Incentive Mechanism for Uncertain Tasks under
Differential Privacy}

\author{Xikun Jiang, Chenhao Ying, Lei Li, Boris D$\ddot{u}$dder, Haiqin Wu, Haiming Jin, Yuan Luo

\IEEEcompsocitemizethanks{
\IEEEcompsocthanksitem Xikun Jiang, Chenhao Ying, Haiming Jin and Yuan Luo are with the Department of Computer Science, Shanghai Jiao Tong University, Shanghai, 200240, China. Xikun Jiang is also with the Department of Computer Science, University of Copenhagen,  Copenhagen, 2100,  Denmark. Chenhao Ying and Yuan Luo are also with Shanghai Jiao Tong University (Wuxi) Blockchain Advanced Research Center, Jiangsu, 214101, China. (E-mails: \{xikunjiang, yingchenhao, jinhaiming, yuanluo\}@sjtu.edu.cn, xikun@di.ku.dk.)
\IEEEcompsocthanksitem Lei Li and Boris D$\ddot{u}$dder are with the Department of Computer Science, University of Copenhagen,  Copenhagen, 2100, Denmark. (E-mail: \{lilei, boris.d\}@di.ku.dk.)
\IEEEcompsocthanksitem
Haiqin Wu is with the Department of Software Engineering Institute, East China Normal University, Shanghai, 200062, China. (E-mail: hqwu@sei.ecnu.edu.cn.)
\IEEEcompsocthanksitem This work is supported in part by the Shanghai Science and Technology Innovation Action Plan 23511100400, in part by the 2023-2024 Open Project of Key Laboratory Ministry of Industry and Information Technology-Blockchain Technology and Data Security 20242216, in part by the National Natural Science Foundation of China (NSFC) under Grants  62372288, and U20A20181.
\IEEEcompsocthanksitem
Corresponding author: Chenhao Ying.
}}

\maketitle

\begin{abstract}
Mobile crowd sensing (MCS) has emerged as an increasingly popular sensing paradigm due to its cost-effectiveness. This approach relies on platforms to outsource tasks to participating workers when prompted by task publishers. Although incentive mechanisms have been devised to foster widespread participation in MCS, most of them focus only on static tasks (i.e., tasks for which the timing and type are known in advance) and do not protect the privacy of worker bids. In a dynamic and resource-constrained environment, tasks are often uncertain (i.e., the platform lacks a priori knowledge about the tasks) and worker bids may be vulnerable to inference attacks. This paper presents an incentive mechanism HERALD*, that takes into account the uncertainty and hidden bids of tasks without real-time constraints. Theoretical analysis reveals that HERALD* satisfies a range of critical criteria, including truthfulness, individual rationality, differential privacy, low computational complexity, and low social cost. These properties are then corroborated through a series of evaluations.
\end{abstract}

\begin{IEEEkeywords}
Uncertain Tasks without Real-time Constraints, Differential Privacy, Mobile Crowd Sensing, Incentive Mechanism.
\end{IEEEkeywords}

\section{Introduction}\label{secintro}
\IEEEPARstart{T}{he} proliferation of mobile devices equipped with advanced processors and numerous sensors (like GPS and microphones), has been a key driver behind the rise of mobile crowd sensing (MCS) as a prominent sensing paradigm. MCS relies on a pool of workers equipped with mobile devices to gather sensory data and has fueled the development of various applications, including smart transportation, traffic management, and IoT. Numerous MCS systems have been developed and implemented accordingly in these and other areas\cite{rdinfocom,fantmc,pantmc,yutmc,liutmc,zhang2021optimizing,yang2023burst,yang2023detfed}.

Fig.~\ref{fig_illustration} provides an illustration of a typical MCS system, in which platforms enlist participating workers to execute tasks upon request from demanders (also referred to as task publishers). The success of the majority of MCS applications hinges on the willingness of an adequate number of mobile workers to partake in the process, thereby ensuring the provision of high-quality services. However, workers may exhibit hesitance towards participating in MCS, since task execution may drain their battery power, storage, computation, and communication resources. Furthermore, participating in MCS may render workers' private information (including their location and bid details) vulnerable to exposure during data collection and exchange. To counterbalance these costs and safeguard their privacy, it is crucial to furnish workers with suitable incentives that do not jeopardize their confidentiality.

\begin{figure}[t]
\centering

    \includegraphics[width=2.5in]{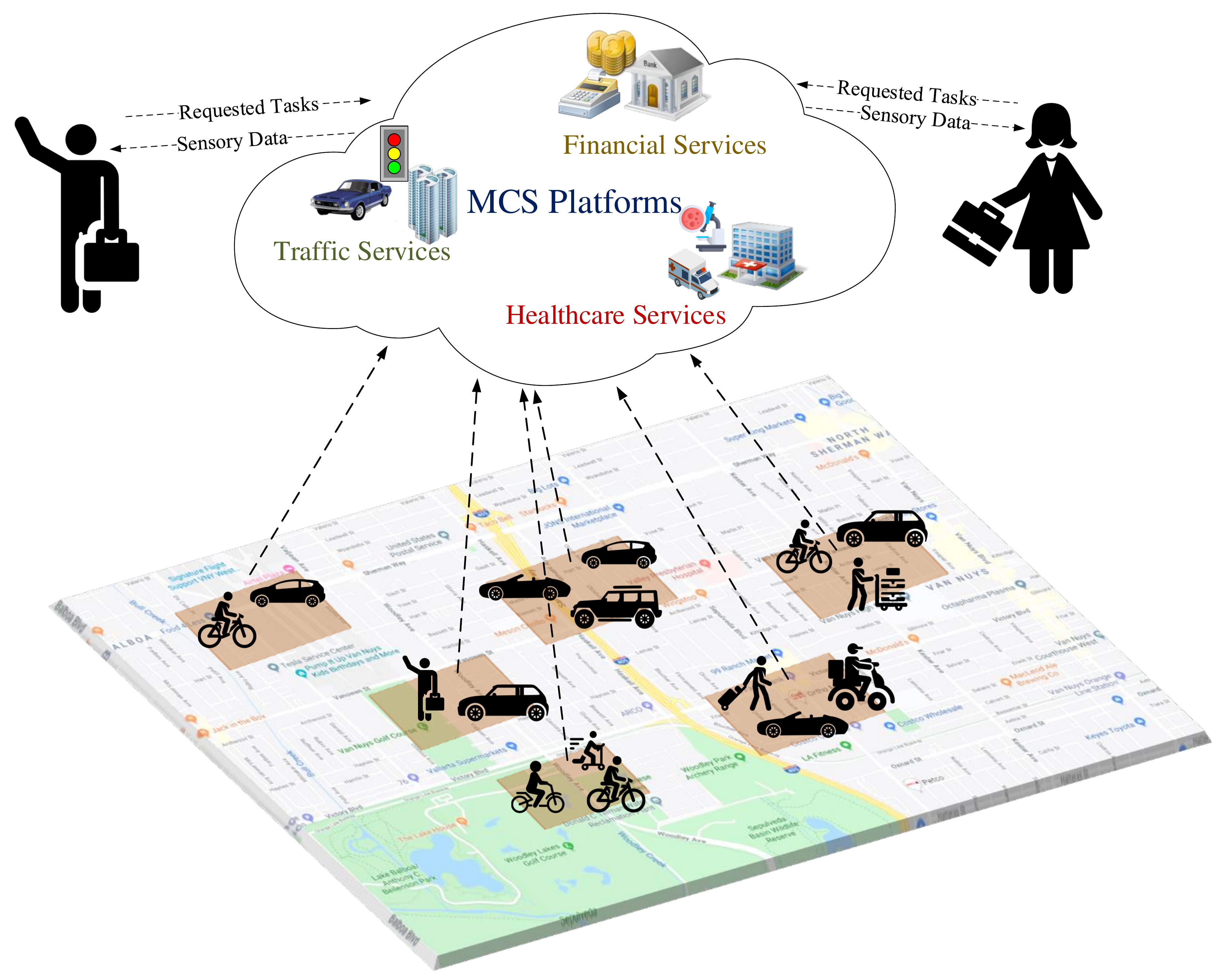}

    \caption{A typical MCS system.}
\label{fig_illustration}
\vspace{-0.2in}
\end{figure}

The importance of incentives has led to the development of numerous mechanisms \cite{zhibohuang2019infocom,ruitingzhou2018tmc,LeiYang2018mobihoc,Cheung2018tmc,Ma_ton2018,Qu_tmc2020,Restuccia2019tmc,wenqiangjin2019infocom,liangwang2018tmc,guo2020fedcrowd,gao2021trustworker,Jin_tmc2019,jiang2021incentive,basik2018fair,Bhattacharjee_tmc2020,Gong_ton2019,liu2021truthful,yidanhu2019infocom,Han_tmc2019} to encourage participation in MCS. However, a considerable number of these mechanisms are grounded on the presumption that tasks are static (i.e., the platform knows the timing and type of tasks in advance). In the real world, MCS tasks are often uncertain due to their unknown arrival time and incomplete information known to the platform. Additionally, these mechanisms do not protect the privacy of worker bids, which may be exposed to potential inference attacks~\cite{bai2017sensor} if published by the platform. The term ``bid'' represents the amount a worker (participant) is willing to accept as compensation for undertaking a task in the mobile crowd-sensing platform. Therefore, there is a desperate need for an incentive mechanism that can handle uncertain tasks and protect worker bid privacy in MCS systems, meanwhile satisfying a set of desired properties such as truthfulness, individual rationality, differential privacy, and low social cost.

In an MCS system, tasks such as collecting information on the number of bends, bifurcations, or roadside shops can be completed in advance, which reduces latency and improves efficiency by allowing the platform to respond immediately when these tasks arrive. These tasks are referred to as non-real-time tasks. However, due to the complexity of the real-world environment, it is often difficult for the platform to predict which tasks will arrive in the future and when they will arrive. These unpredictable tasks are referred to as \emph{uncertain tasks}.

Designing a suitable incentive mechanism for MCS systems, which are susceptible to both uncertain tasks without real-time constraints and inference attacks, poses a considerable challenge. To tackle this issue, we consider a scenario where the arrival of uncertain tasks without real-time constraints is governed by a probability distribution and employ the exponential mechanism, a technique from differential privacy, to safeguard the privacy of worker bids. Our proposed solution, HERALD*\footnote{The name HERALD* is from incentive mec\underline{H}anism for unc\underline{E}\underline{R}t\underline{A}in tasks without real-time constraints under differential privacy in mobi\underline{L}e crow\underline{D} sensing}, which satisfies truthfulness, individual rationality, and differential privacy, while also having low computational complexity and social cost. The main contributions of this paper are:

\begin{itemize}
\item
\emph{Mechanism:} In contrast to previous approaches \cite{zhibohuang2019infocom,ruitingzhou2018tmc,LeiYang2018mobihoc,Cheung2018tmc,Ma_ton2018,Qu_tmc2020,Restuccia2019tmc,wenqiangjin2019infocom,liangwang2018tmc,guo2020fedcrowd,gao2021trustworker,Jin_tmc2019,jiang2021incentive,basik2018fair,Bhattacharjee_tmc2020,Gong_ton2019,liu2021truthful,yidanhu2019infocom,Han_tmc2019}, our work introduces a new incentive mechanism named HERALD*, which integrates differential privacy. Specifically, HERALD* is tailored for uncertain tasks without real-time constraints that are expected to arrive based on a probability distribution (in our study, we consider the uniform distribution), enabling the platform to gather sensory data in advance. Additionally, we employ the exponential mechanism to protect the privacy of worker bids.

\item
\emph{Desirable Properties:} HERALD* can effectively encourage worker participation and achieves a set of desirable properties, such as truthfulness, individual rationality, differential privacy, low computational complexity, and low social cost. Unlike other incentive mechanisms such as those proposed in \cite{zhibohuang2019infocom,ruitingzhou2018tmc,LeiYang2018mobihoc,Cheung2018tmc,Ma_ton2018} for traditional MCS, HERALD* is specifically designed for uncertain tasks without real-time constraints and is not limited to collecting large amounts of sensory data. Furthermore, we demonstrate that HERALD* has a competitive ratio of $\mathcal{O}(\ln l n)$ in terms of expected social cost, where $l$ and $n$ represent the number of task subsets and tasks published in advance, respectively. We also prove that HERALD* preserves $\frac{\epsilon l}{2}$-differential privacy for both linear and logarithmic score functions, where $\epsilon > 0$ is a constant.

\item
\emph{Evaluations:} In addition to HERALD*'s desirable properties, we performed comprehensive simulations to verify its efficacy. Our findings indicate that HERALD* outperforms current approaches by exhibiting a lower anticipated social cost and total payment while simultaneously providing differential privacy.
\end{itemize}

The subsequent sections of this paper are structured as follows: Section II provides a discussion of the related literature, while Section III presents an introduction to the preliminaries. Section IV outlines the design and theoretical analysis of HERALD*. In Section V, extensive simulations are conducted to validate the properties of the proposed mechanism. Lastly, Section VI concludes the paper.

\section{Related Work}
Numerous incentive mechanisms \cite{zhibohuang2019infocom,ruitingzhou2018tmc,LeiYang2018mobihoc,Cheung2018tmc,Ma_ton2018,Qu_tmc2020,Restuccia2019tmc,wenqiangjin2019infocom,liangwang2018tmc,guo2020fedcrowd,gao2021trustworker,Jin_tmc2019,jiang2021incentive,basik2018fair,Bhattacharjee_tmc2020,Gong_ton2019,liu2021truthful,yidanhu2019infocom,Han_tmc2019} have been proposed for MCS systems since attracting a considerable number of workers is crucial. In addition to truthfulness and individual rationality, these mechanisms often aim to ensure the benefits of both workers and the platform.

The authors of \cite{zhibohuang2019infocom,ruitingzhou2018tmc,LeiYang2018mobihoc} proposed mechanisms to minimize the social cost, while \cite{Cheung2018tmc,Ma_ton2018} aimed to maximize the platform's profit. In contrast, \cite{Qu_tmc2020,Restuccia2019tmc,wenqiangjin2019infocom,liangwang2018tmc} focused on minimizing the platform's payment, and \cite{Jin_tmc2019,jiang2021incentive} designed mechanisms to maximize social welfare. A novel crowdsourcing assignment strategy proposed in \cite{basik2018fair} considered fair task allocation for workers. In addition to the above objectives, there have been efforts to achieve other objectives as well. For instance, Bhattacharjee \emph{et al.} in \cite{Bhattacharjee_tmc2020} incentivized workers to act honestly by evaluating the quantity and quality of their data. Gong \emph{et al.} in \cite{Gong_ton2019} introduced an incentive mechanism to encourage workers to submit high-quality data, while Liu \emph{et al.} in \cite{liu2021truthful} addressed the problem of multi-resource allocation by devising a truthful double auction mechanism.

Moreover, several studies investigated privacy-preserving methods in mobile crowdsourcing. For example, Lin \emph{et al.} \cite{lin2016bidguard} proposed a general privacy-preserving framework for incentivizing crowdsensing using two score functions. Hu \emph{et al.} \cite{yidanhu2019infocom} developed a privacy-preserving incentive mechanism for dynamic spectrum sharing crowdsensing, while Han \emph{et al.} \cite{Han_tmc2019} focused on privacy-preserving in budget-limited crowdsensing. Yang \emph{et al.} \cite{yang2020secure} incorporated both additive secret sharing and local differential privacy technologies. Wei \emph{et al.} \cite{wei2019differential} investigated location privacy-preserving in spatial crowdsourcing. Some studies have also used cryptography techniques \cite{wu2022privacy,shu2018privacy}, or blockchain \cite{guo2020fedcrowd,gao2021trustworker,wu2021privacyaware} to protect privacy in mobile crowdsourcing, but they do not consider the strategic behavior of the participants.

Most previously mentioned works have focused on static sensing tasks, where the platform has complete knowledge of the task information beforehand. However, in realistic environments with resource constraints, sensing tasks are often uncertain, meaning that the platform has incomplete knowledge of the task information. In contrast to these previous works, this paper addresses the issue of uncertain tasks while also ensuring the bidding privacy of participants. 

\section{Preliminaries}\label{secpre}
In this section, we provide an overview of the system and discuss the objectives that guided its design.

\subsection{System Overview}\label{subsecsys}
In this study, we focus on a mobile crowdsourcing system (MCS) consisting of a cloud-based platform and a group of participating workers (represented as $\mathcal{W}=\{1,2,\ldots,m\}$). The platform is assumed to have a priori knowledge of a set of $n$ sensing tasks, $\mathcal{T}=\{\tau_{1},\ldots,\tau_{n}\}$, and all future requested tasks will belong to $\mathcal{T}$. Note that index $n$ is a variable. The platform divides these $n$ tasks into $l$ task subsets, $\mathcal{Y}=\{\Gamma_{1},\Gamma_{2},\ldots,\Gamma_{l}\}$, such that $\bigcup^{l}_{j=1}\Gamma_{j}=\mathcal{T}$. The constraint $l<mn$ is set to ensure that the count of potential task subsets is bounded by the collective interactions between the workers and tasks. Additionally, the number of tasks within each subset $\Gamma_{j}$ is randomized. To ensure task quality and prevent monopolization, we require that each task $\tau_{i}\in\mathcal{T}$ must be covered by no less than two task subsets. This condition is practical since the platform typically possesses some prior knowledge about the tasks that need to be performed. For instance, in a traffic monitoring application, the task set may comprise the number of all the intersections on a particular road. This task set remains static over time and has no real-time constraints. A similar application could be the monitoring of road curves. The system framework is depicted in Fig.~\ref{fig_workflow}, and its operation is elaborated below.

\begin{figure}[t]
\centering
\centering
 \includegraphics[width=2.5in]{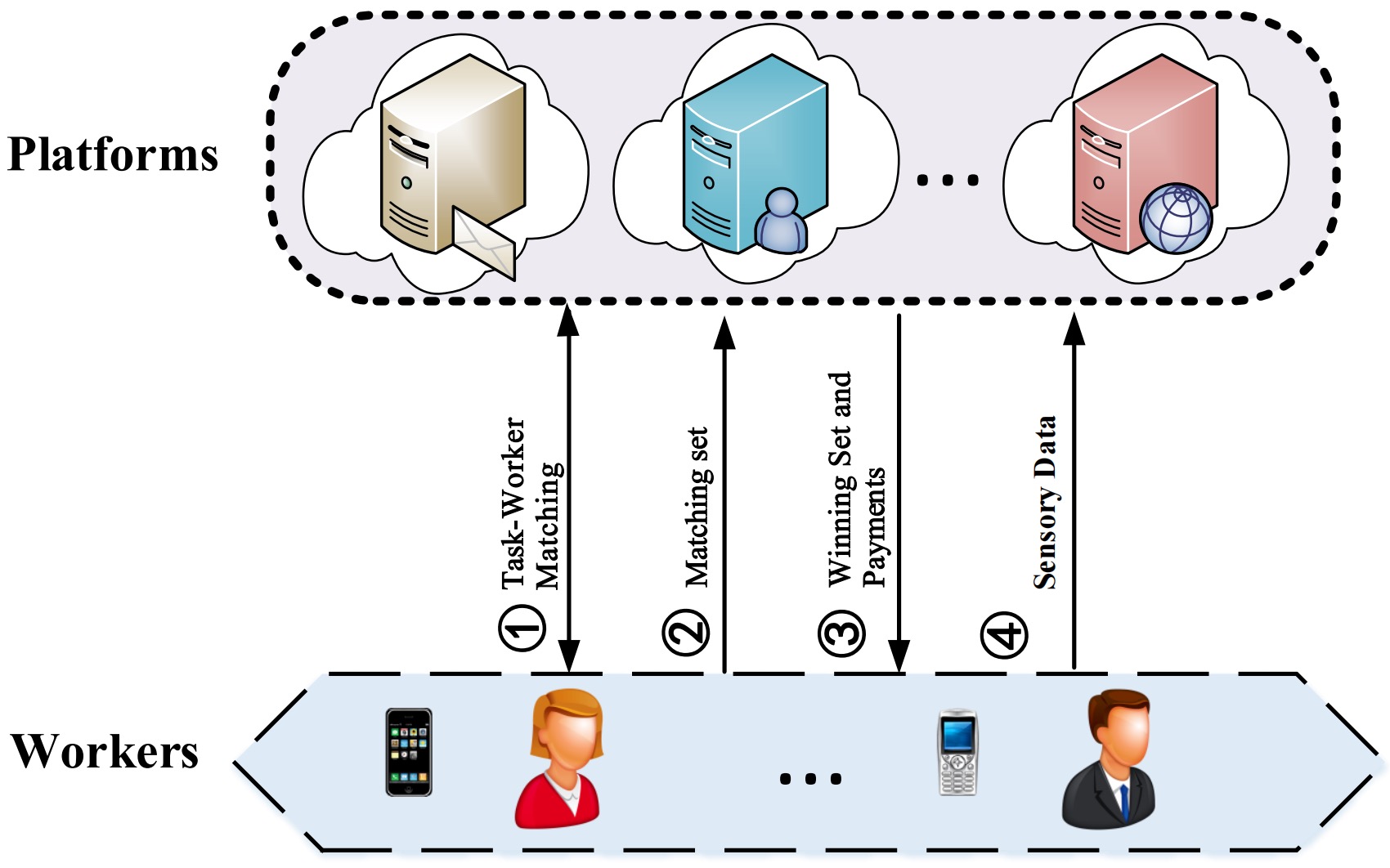}
\caption{Framework of HERALD*.}\label{fig_workflow}
\vspace{-0.2in}
\end{figure}

\textbf{Incentive Mechanism for Uncertain Tasks without Real-time Constraints:} As illustrated in Fig.~\ref{fig_workflow}, before the actually requested tasks arrive, the platform first matches workers to task subsets in a task-worker matching phase (step \textcircled{1}). For each task subset $\Gamma_{j} \in \mathcal{Y}$, each worker $i$ submits a bid $b_i$. We assume that the probability that task subset $\Gamma_{j}$ is matched with worker $i$ is $Pr(b_i)$. Note that in each round of task-worker matching, each worker can only participate once, and in the $l$ rounds of matching, each worker $i$ only has one unique bid $b_i$. This means that for each task subset $\Gamma_{j}$, the probability that worker $i$ is matched is $Pr(b_i)$. The platform randomly assigns the task subset to the workers according to this probability distribution. To ensure task quality and prevent monopolization, we require that each task $\tau_{i}\in\mathcal{T}$ is matched by at least two different workers. After the task-worker matching phase, the platform receives the matching set $\mathcal{P}$, which consists of matching pairs $(\Gamma_{j}, b_i)$ for $\Gamma_{j}\in \mathcal{Y}$ and $i\in W$, where $b_{i}$ is worker $i$'s bid for executing task subset $\Gamma_{j}$ (step \textcircled{2}). Upon receiving a matching set $\mathcal{P}$, the platform proceeds to determine the winning set $\mathcal{S}$ and the corresponding payment $p_{i}$ to each winning worker $i$. Among them, the winning set $\mathcal{S}$ consists of winning matching pairs, and the workers contained in the winning matching pairs are selected to perform tasks based on their bids and matching task sets. As indicated by step \textcircled{3}. Additionally, the platform collects sensory data from the winning workers to enable prompt responses to future requests, as per step \textcircled{4}. It is worth noting that the platform collects sensory data for tasks in $\mathcal{T}$ in advance, as it lacks any knowledge about future tasks, which are assumed to arrive according to a probability distribution.

In this mechanism, a worker who is not selected to execute any tasks, as a ``loser'', receives zero payment. We use the notation $\overrightarrow{p}=(p_{1},\ldots,p_{m})$ to represent the payment profile of the workers, which is initialized to be zero. If we denote the cost of worker $i$ as $c_{i}$ in HERALD*, we can define the utility of the worker as 

\begin{footnotesize}
\begin{equation}\label{utility}
\begin{split}
u_{i}= \left\{
             \begin{array}{ll}
             {p_{i}-c_{i}}&\text{, if worker $i$ wins,}\\
             {0}&\text{, otherwise.}
             \end{array}
        \right.
\end{split}
\end{equation}
\end{footnotesize}

Where the cost refers to the anticipated expenditure workers estimate before starting a task, including time commitment, battery usage, storage needs, computing demands, and communication resources~\cite{Casella2022}. These factors are crucial considerations for workers in deciding to participate in the MCS ecosystem. Therefore, the platform must provide compensation that adequately offsets these incurred costs, ensuring workers receive non-negative benefits. Without loss of generality, we assume that the bid $b_{i}$ of each worker $i$ is constrained within the range of $[b_{min},b_{max}]$, where $b_{min}$ is normalized to $1$ and $b_{max}$ is a constant. We use the notation $\overrightarrow{b}=(b_{1},\ldots,b_{m})$ to represent the bid profile of the workers for convenience. Additionally, it is assumed that for each worker $i$ who has a specific subset of tasks, denoted as $\Gamma_{i}$, there exist other workers $j$ who have task subsets $\Gamma_{j}$ that include the tasks of worker $i$, i.e., $\Gamma_{i}\subseteq\cup_{j}\Gamma_{j}$.

\subsection{Design Objectives}
\label{subsecdes}
The main objective of this manuscript is to establish that HERALD* exhibits the subsequent advantageous features. Given that workers may exhibit selfish or strategic behavior, it is plausible that any worker $i$ may present a bid $b_{i}$ that differs from their actual cost $c_{i}$ of performing all the tasks in $\Gamma_{i}$. Consequently, we aim to develop an incentive mechanism that satisfies truthfulness, which is defined as follows.
\begin{definition}[Truthfulness]\label{IC}
An incentive mechanism is truthful if for any worker $i\in\mathcal{W}$, his/her utility is maximized when bidding his/her true cost $c_{i}$.
\end{definition}

According to Definition \ref{IC}, our goal is to ensure that workers make truthful bids to the platform. In addition to truthfulness, we also strive to achieve another desirable property known as individual rationality, which is defined below:

\begin{definition}[Individual Rationality]\label{IR}
An incentive mechanism is individually rational if for any worker $i\in\mathcal{W}$, his/her utility $u_{i}$ satisfies $u_{i}\geq0$.
\end{definition}

Nonetheless, if the platform discloses the bidding results that contain the winning bidders and their rewards directly, it may expose the privacy of the participants to inference attacks, thereby impeding the advantages of the platform. Since any changes in workers' bids can considerably influence the ultimate bidding results, particularly the payments, an adversarial worker may deduce the bids of other workers based on the distinct payments received in various auctions. 

\textbf{Adversary Model Analysis:} The driving force behind the implementation of robust privacy protection mechanisms within HERALD* is the imperative to counter the looming threat of malicious worker inference attacks. These attacks possess the potential to inflict substantial disruption upon the operation and integrity of the Mobile Crowd Sensing (MCS) platform. Malicious workers can exploit the bids of fellow participants, thereby gaining insights into their strategies and subsequently employing strategic bid manipulations. If left unaddressed, these manipulations can wreak havoc on the core operations of the platform. Within the MCS framework, the bids submitted by workers play a pivotal role in the decisive winning selection phase. Malicious workers, through the manipulation of bids, can exert undue influence on the outcomes of this phase, thereby imperiling fairness and introducing imbalances into the selection process. The repercussions of malicious worker inference attacks extend even further, reaching into the critical domain of payment determination. Bid manipulation, if unchecked, can culminate in the unjust distribution of compensation, directly impacting the interests of both the platform and its users. This misalignment of incentives, in turn, has the potential to lead to a reduction in participation and the erosion of trust within the MCS ecosystem.

To protect against this inference attack and ensure the privacy of the workers' bids, it is necessary to devise a mechanism that fulfills differential privacy, which is defined as follows.

\begin{definition}[Differential Privacy\cite{dwork2014algorithmic}]\label{DP} A randomized mechanism $M$ preserves $(\epsilon,\delta)$-differential privacy if for any two input sets $A$ and $B$ with a single input difference, and for any set of outputs $\mathcal{O} \in Range(M)$,
\begin{footnotesize}
\begin{equation}\label{differential}
Pr[M(A)\in \mathcal{O}] \le exp(\epsilon) \times Pr[M(B)\in \mathcal{O}] + \delta.
\end{equation}
\end{footnotesize}
If $\delta=0$, we say that $M$ preserves $\epsilon$-differential privacy.
\end{definition}

Let $Pr(b_i)$ represent the probability that worker $i$ is matched with any task subset $\Gamma_{j}\in \mathcal{Y}$ when their bid is $b_i$. To safeguard the privacy of the workers' bids, it is necessary to decrease the impact of the distinct bids of workers on the eventual bidding outcome. To do this, we introduce differential privacy, specifically using the exponential mechanism which is defined as follows.

\begin{definition}[The Exponential Mechanism~\cite{dwork2014algorithmic}]\label{EP} The exponential mechanism $\mathcal{M}_E(x, u, R)$ selects and outputs an element $r \in R$ with probability proportional to $exp(\frac{\epsilon u(x,r)}{2\Delta u})$, i.e., $Pr[\mathcal{M}_E(x, u, R)=r] \propto exp(\frac{\epsilon u(x,r)}{2\Delta u})$, where $x$ is the input set and $u$ is a utility function that maps input/output pairs to utility scores, $\Delta u=max_{r \in R} max_{x,y:{||x-y||}_1\le 1} |u(x,r)-u(y,r)|$ is the sensitivity of the utility function $u$, and $\epsilon $ is a small constant.
\end{definition}

We have integrated the above exponential mechanism into HERALD* to safeguard the input, which consists of the task set and bid, while the output is the worker-task matching result. This ensures that changes in the input do not alter the output, preserving the differential privacy of bids. This method effectively shields the bids from being identified by adversaries, particularly during critical phases such as winner selection and payment determination, which rely on these protected bids. The following theorem is derived from the exponential mechanism mentioned above.

\begin{theorem}[\cite{dwork2014algorithmic}]\label{EPDP}
The exponential mechanism $\mathcal{M}_E(x, u, R)$ preserves $(\epsilon, 0)$-differential privacy.
\end{theorem}

To ensure that changes in workers' bids do not significantly affect the final bidding results, the exponential mechanism is implemented to achieve differential privacy in HERALD*. This mechanism guarantees that any changes in the input data remain hidden, thus safeguarding the privacy of individual bids. Moreover, this protection extends seamlessly to subsequent phases,  as these phases are entirely based on the protected bids. This makes it difficult for malicious workers to infer the bidding details of other workers based on the final bidding results. Similar to previous work \cite{lin2016bidguard,yidanhu2019infocom,Han_tmc2019,wei2019differential}, we incorporate randomization into the incentive mechanism's outcome to achieve differential privacy.

Apart from the above objectives, we also strive for HERALD* to possess a low anticipated social cost, taking into account the probability distribution of the tasks in $\mathcal{T}$ arrive. We assess the competitive ratio of the system's expected social cost as a measure to accomplish this aim.

\begin{definition}[Competitive Ratio on Expected Social Cost]\label{KEP}
Suppose the tasks in the sensing task set $\mathcal{T}$ arrive in a probability distribution, for any set $\mathcal{A}$ of $k$ tasks that may arrive at the same time from $\mathcal{T}$, let $\mathcal{S}(\mathcal{A},W)$ refer to the winning set chosen by the mechanism such that $\mathcal{A}\subseteq\cup_{(\Gamma_j,c_i)\in\mathcal{S}(\mathcal{A},W)}\Gamma_{j}$ and $\Gamma_{j}\cap\mathcal{A}\neq\emptyset$ for $\forall i\in\mathcal{S}(\mathcal{A},W)$, $C(\mathcal{S}(\mathcal{A},W))=\sum_{(\Gamma_j,c_i)\in\mathcal{S}(\mathcal{A},W)}c_{i}$ be the corresponding \textbf{social cost}, and $C_{\mathcal{OPT}}(\mathcal{A},W)$ be the minimum social cost of a requested task set $\mathcal{A}$, respectively. The competitive ratio on expected social cost is defined as $\max_{k}\mathbb{E}_{\mathcal{P}\in \digamma}\big [\mathbb{E}_{\mathcal{A}\subseteq\mathcal{T}}[C(\mathcal{S}(\mathcal{A},W))]/\mathbb{E}_{\mathcal{A}\subseteq\mathcal{T}}[C_{\mathcal{OPT}}(\mathcal{A},W)]\big ]$, where $\mathbb{E}_{\mathcal{P}\in \digamma}[\cdot]$ is the expectation over all matching set results $\digamma$ and $\mathbb{E}_{\mathcal{A}\subseteq\mathcal{T}}[\cdot]$ is the expectation over all sets of possibly $k$ arriving tasks in the future.
\end{definition}

It is worth noting that a worker may be included multiple times in the winning set, having the same cost but with a different subset of tasks since each worker can match with several subsets of tasks. Additionally, it is possible for some tasks in the task set $\mathcal{A}$ to be identical. Thus, when referring to $\mathcal{A}\subseteq\mathcal{T}$ in the evaluation of the competitive ratio, we imply that every task in $\mathcal{A}$ also belongs to $\mathcal{T}$, as the tasks in $\mathcal{T}$ are distinct. Moreover, the expectation $\mathbb{E}_{\mathcal{A}\subseteq\mathcal{T}}[\cdot]$ is determined by the variability of the set $\mathcal{A}$ of $k$ requested tasks. For convenience of notation, the subscript $\mathcal{A}\subseteq\mathcal{T}$ is omitted in the rest of this paper, and this expectation is expressed as $\mathbb{E}[\cdot]$.

Ultimately, we strive for HERALD* to be computationally efficient, and we define this objective as follows.
\begin{definition}\label{ce}
An incentive mechanism is computationally efficient if it can be executed within polynomial time.
\end{definition}

In essence, our aims are to guarantee that the proposed mechanism is both truthful and individually rational, as well as being differentially private, while also having a small social cost and computational complexity.

\section{Incentive Mechanism for Uncertain Tasks without real-time constraints}\label{sec-UC}
In this section, we introduce an incentive mechanism designed for uncertain tasks without real-time constraints, and we demonstrate that our mechanism satisfies the properties of truthfulness, individual rationality, and differential privacy. Additionally, we examine the competitive ratios on the expected social cost of the mechanism, as stated in Theorem \ref{expected_ratio}. Furthermore, we evaluate the computational complexity of the mechanism, as presented in Proposition~\ref{complexity-A}, in addition to the aforementioned properties.

\subsection{Design Rationale}
The design of an incentive mechanism must take into account the possibility of tasks arriving simultaneously, as this can affect the mechanism's construction. In the offline scenario, all tasks arrive simultaneously, and the platform has prior knowledge of all tasks. However, in the case of uncertain tasks, the number of tasks arriving simultaneously is also uncertain, and the probability distribution of the number of arriving tasks is different. To clarify this concept, consider the following straightforward example.

\begin{example}
In this example, the platform possesses a sensing task set $\mathcal{T}=\{\tau_{1},\tau_{2},\tau_{3}\}$ containing three tasks, each of which has a probability of $\frac{1}{3}$ to arrive in the future. As a result, the task arrivals follow a uniform distribution. If a single task is set to arrive in the future, it could be $\tau_{1}$, $\tau_{2}$ or $\tau_{3}$ with the same probability $\frac{1}{3}$. While, if two tasks arrive simultaneously in the future, they may be $\{\tau_{1},\tau_{1}\}$, $\{\tau_{2},\tau_{2}\}$ and $\{\tau_{3},\tau_{3}\}$ with the same probability $\frac{1}{9}$, and may be $\{\tau_{1},\tau_{2}\}$, $\{\tau_{1},\tau_{3}\}$ and $\{\tau_{2},\tau_{3}\}$ with the same probability $\frac{2}{9}$. Furthermore, if three tasks simultaneously arrive in the future, they may be $\{\tau_{1},\tau_{2},\tau_{3}\}$ with probability $\frac{2}{9}$; $\{\tau_{1},\tau_{1},\tau_{1}\}$, $\{\tau_{2},\tau_{2},\tau_{2}\}$ and $\{\tau_{3},\tau_{3},\tau_{3}\}$ with the same probability $\frac{1}{27}$; and may be $\{\tau_{1},\tau_{2},\tau_{2}\}$, $\{\tau_{1},\tau_{3},\tau_{3}\}$,
$\{\tau_{1},\tau_{1},\tau_{2}\}$,
$\{\tau_{1},\tau_{1},\tau_{3}\}$,
$\{\tau_{2},\tau_{2},\tau_{3}\}$, and
$\{\tau_{2},\tau_{3},\tau_{3}\}$ with the same probability $\frac{1}{9}$.
\end{example}

As demonstrated in the example, our approach differs from existing works. Rather than assuming a fixed number of tasks arriving simultaneously, we propose HERALD*, an adaptive incentive mechanism that adjusts to varying numbers of tasks as they arrive. The first phase of HERALD* involves matching workers to all possible task subsets, before the actual task requests are received. Subsequently, an estimated number of tasks arriving simultaneously is inputted into HERALD*. Based on the different input numbers, HERALD* will produce varying outcomes for both winning selection and payment determination.

\subsection{Design Details}\label{subsecDD-UC}
Within this section, we shall furnish an in-depth account of the operational mechanics of HERALD*. This mechanism consists of three phases: the task-worker matching phase (Alg.~\ref{alg1}), the winning selection phase (Alg.~\ref{alg2}), and the payment determination phase (Alg.~\ref{alg3}).

 \begin{algorithm}[ht]
            \caption{HERALD*: Task-Worker Matching}
            \begin{algorithmic}[1]
            \renewcommand{\algorithmicrequire}{\textbf{Input:}}
            \renewcommand{\algorithmicensure}{\textbf{Output:}}
            \label{alg1}
            \REQUIRE The task set $\mathcal{T}$, the set of task subsets $\mathcal{Y}=\{\Gamma_{1},\Gamma_{2},\ldots,\Gamma_{l}\}$ such that  $\bigcup^{l}_{j=1}\Gamma_{j}=\mathcal{T}$, worker set $\mathcal{W}$, each worker's bid, a score function $u$ and its sensitivity $\Delta u$, a small constant $\epsilon$.
            \ENSURE Matching set $\mathcal{P}$.
            
            \STATE $\mathcal{P}\leftarrow\emptyset$

            \STATE $S = \sum_{i\in \mathcal{W}} exp(\frac{\epsilon u(b_i)}{2\Delta u})$
            
            \FOR{each worker $i \in \mathcal{W}$}
            
            \STATE Calculate the probability $Pr(b_i)$ that worker $i$ with his/her unique bid $b_i$ is matched with each task subset $\Gamma_{j} \in \mathcal{Y}$: $Pr(b_i) = \frac{exp(\frac{\epsilon u(b_i)}{2\Delta u})}{S}$
            \ENDFOR
            \FOR{each task subset $\Gamma_{j} \in \mathcal{Y}$}
            
            \STATE Select the worker $i$ randomly according to the computed probability distribution to match the task subset $\Gamma_{j}$

            \STATE $\mathcal{P}\cup\{
            (\Gamma_{j}, b_i)\}$
            
            \ENDFOR
        \RETURN The matching set $\mathcal{P}$ of the task subset $\Gamma_{j}$ and its matched worker $i$'s bid.
    \end{algorithmic} 
    \end{algorithm}

 \begin{algorithm}[ht]
            \caption{HERALD*: Winning  Selection}
            \renewcommand{\algorithmicrequire}{\textbf{Input:}}
             \renewcommand{\algorithmicensure}{\textbf{Output:}}
             \begin{algorithmic}[1]
            \label{alg2}
            \REQUIRE The task set $\mathcal{T}$, matching set $\mathcal{P}$, each worker's bid, the number $k$ of tasks arriving simultaneously.
            \ENSURE The winning set $\mathcal{S}$.
            \STATE $\mathcal{S}\leftarrow\emptyset$\;

            \STATE Calculate the selection threshold $T=64\mathbb{E}[C_{\mathcal{OPT}}(\mathcal{A},W)]$, where $\mathcal{A}$ is the set of $k$ possibly simultaneously arriving tasks from the sensing task set $\mathcal{T}$

            \WHILE {$\mathcal{T}\neq\emptyset$}{
            \FOR {each matching pair  $(\Gamma_{j}, b_i)\in \mathcal{P}$}
            \STATE Calculate the cost effectiveness (CF) $\frac{b_{i}}{|\Gamma_{j}\cap\mathcal{T}|}$, where $b_i$ is the bid of worker $i$ for executing the task subset $\Gamma_{j}$
            \ENDFOR
            \\ \textit{Type I Selection:}
            \IF {$\exists (\Gamma_{j}, b_i)\in \mathcal{P}$, s.t $\frac{b_{i}}{|\Gamma_{j}\cap\mathcal{T}|}\leq \frac{T}{|\mathcal{T}|}$}
            
            \STATE Among the task subsets whose CFs are less than $\frac{T}{|\mathcal{T}|}$, the worker $i\in\mathcal{W}$ that matches the task subset $\Gamma_{j}$ with the minimum CF value is selected as the winner 
            \\ \textit{Type II Selection:}
            \ELSE
            \STATE Among the workers matched by all task subsets, the worker $i\in\mathcal{W}$ whose bid is the lowest and whose task subset $\Gamma_{j}$ she matched contains at least one undiscovered task is selected as the winner
            \ENDIF
            \STATE $\mathcal{S}\leftarrow\mathcal{S}\cup\{(\Gamma_{j},b_i)\}$
            
            \STATE $\mathcal{P}\leftarrow\mathcal{P}\backslash\{(\Gamma_{j}, b_i)\}$
            \STATE $\mathcal{T}\leftarrow\mathcal{T}\backslash\Gamma_{j}$\;
            }
            \ENDWHILE
        \RETURN $\mathcal{S}$.
    \end{algorithmic}
    \end{algorithm}

 \begin{algorithm}[!t]
            \caption{HERALD*: Payment Determination}
            \begin{algorithmic}[1]
            \renewcommand{\algorithmicrequire}{\textbf{Input:}}
            \renewcommand{\algorithmicensure}{\textbf{Output:}}
            \label{alg3}
            \REQUIRE The worker set $W$, winning set $S$, each worker's bid.
            \ENSURE The payment $\overrightarrow{p}$.
            \FOR{each winning pair $(\Gamma_{j}, b_i)\in S$}
            \STATE Define a \emph{copy set} $\mathcal{T}_{j}\leftarrow\Gamma_{j}$
            
            \STATE Build a \emph{covering set} $\mathcal{W}_{i}=\{\ell|\forall \ell\in\mathcal{W}\backslash\{i\},\ \Gamma_{\ell}\cap\mathcal{T}_{j}\neq\emptyset\}$
            
            \STATE Define a \emph{replaced set} $\mathcal{R}_{i}\leftarrow\emptyset$\;
            \WHILE{$\mathcal{T}_{j}\neq\emptyset$}
            \STATE Choose a worker $\ell\in\mathcal{W}_{i}$ whose matched task subset has the minimum CF
            
            \STATE $\mathcal{R}_{i}\leftarrow \mathcal{R}_{i}\cup\{\ell\}$
            \STATE $\mathcal{T}_{j}\leftarrow\mathcal{T}_{j}\backslash\Gamma_{\ell}$
            \ENDWHILE
            \STATE $p_{i}=p_{i} + \max\{b_{i},p_{\mathcal{R}_{i}}\}$ for $p_{\mathcal{R}_{i}}=\sum_{\ell\in\mathcal{R}_{i}}b_{\ell}$
            \ENDFOR
        \RETURN $\overrightarrow{p}$.
    \end{algorithmic} 
    \end{algorithm}

\textbf{Task-Worker Matching Phase:}
Before the real requested tasks arrive, the platform first matches workers for all task subsets. We refer to this as the task-worker matching phase. As can be seen in Alg.~\ref{alg1}, for each task subset $\Gamma_{j} \in \mathcal{Y}$, the probability that worker $i$ with the unique bid $b_i$ is matched is $Pr(b_i)$ which is calculated in Lines 2-4. In particular, as we introduced in definition \ref{EP}, we employ the exponential mechanism to achieve the differential privacy of the bidding results. Thus we set $Pr(b_i)\propto exp(\frac{\epsilon u(b_i,r)}{2\Delta u})$, where $r \in R$ is the output and $u$ is a score function that maps input/output pairs to utility scores, $\Delta u=max_{r \in R} max_{x,y:{||x-y||}_1\le 1} |u(x,r)-u(y,r)|$ is the sensitivity of the utility function $u$, and $\epsilon $ is a small constant. For each round of matching shown in Lines 5-8, the platform eventually assigns the task subset to the worker randomly according to the computed probability distribution. To ensure task quality and prevent monopolization, we require that each task $\tau_{i}\in\mathcal{T}$ is matched by at least two different workers.

Upon completion of the task-worker matching phase, the platform proceeds to identify the winning set $\mathcal{S}$ and the corresponding payment $p_{i}$ for each winning worker $i$ based on the submitted matching set $\mathcal{P}$. We refer to these as the winning selection phase and payment determination phases, which are illustrated in Alg.~\ref{alg2} and Alg.~\ref{alg3}, respectively.

\textbf{Winning Selection Phase:}
As can be seen in Alg.~\ref{alg2}, in order to obtain the sensory data of the uncertain tasks, we define a \emph{selection threshold} (ST) $T\geq0$ (Line 2) that remains unchanged throughout the implementation process of HERALD*. There are two distinct forms of selection employed in HERALD* during each iteration: \emph{Type I Selection} and \emph{Type II Selection}, which are chosen by cost-effectiveness (CF) (Lines 4-6) defined as follows.

\begin{definition}[Cost-effectiveness]\label{CF}
Let $\mathcal{T}$ be the set of tasks whose sensory data is not collected at the beginning of an iteration in HERALD*. For each matching pair $(\Gamma_{j}, b_{i})$, the CF of a worker $i$ during each iteration is defined as $\frac{b_{i}}{|\Gamma_{j}\cap\mathcal{T}|}$.
\end{definition}

\begin{itemize}
\item
\emph{Type I Selection (Lines 7-8):} If there exist any task subsets with cost-effectiveness (CF) values equal to or lower than $\frac{T}{|\mathcal{T}|}$, the platform will choose the worker who matches the task subset with the lowest CF as the winner.

\item
\emph{Type II Selection (Lines 9-10):} If the CFs of all task subsets exceed $\frac{T}{|\mathcal{T}|}$, the platform will choose a winning worker from among those matched to all task subsets. This worker must possess the lowest bid and have a task subset that includes at least one unassigned task in $\mathcal{T}$.
\end{itemize}

Subsequently, the matching pair that was selected as per the aforementioned criteria are incorporated into the winning set $\mathcal{S}$ (Line 12) and deleted from the matching set $\mathcal{P}$ (Line 13), and the corresponding tasks are also deleted from the task set $\mathcal{T}$ (Line 14). The platform then proceeds to calculate the payment for each winner based on the members of the winning set.

\textbf{Payment Determination Phase:}
As shown in Alg.~\ref{alg3}, for each winning pair $(\Gamma_{j},b_i)\in\mathcal{S}$, the platform defines a \emph{copy set} $\mathcal{T}_{j}=\Gamma_{j}$ and builds a \emph{covering set} $\mathcal{W}_{i}=\{\ell|\forall \ell\in\mathcal{W}\backslash\{i\},\ \Gamma_{\ell}\cap\mathcal{T}_{j}\neq\emptyset\}$, see Line 2 and Line 3 separately. This means, the intersection of the task set matched by workers in \emph{covering set} $\mathcal{W}_{i}$ and the task set matched by worker $i$ is not $0$. It then derives a \emph{replaced set} denoted as $\mathcal{R}_{i}$ consisting of workers in $\mathcal{W}_{i}$ with the least CFs in each iteration such that $\Gamma_{j}\subseteq\cup_{\ell\in\mathcal{R}_{i}}\Gamma_{\ell}$. As shown in Lines 4-9, worker $i$'s task set $\Gamma_{j}$ can be replaced by workers in the \emph{replaced set} $\mathcal{R}_{i}$. Since each worker can match multiple task sets at the same time to form multiple matching pairs, among the multiple matching pairs composed of worker $i$, more than one matching pair may become the winning pair, so the total reward of worker $i$ is equal to the cumulative sum of the rewards of all winning pairs he formed. For each winning pair, worker $i$'s reward is the maximum value between $b_i$ and the total bids of workers in the \emph{replaced set} $\mathcal{R}_{i}$. Therefore, as shown in Line 10, the payment to winner $i$ is $p_{i}=p_{i}+\max\{b_{i},p_{\mathcal{R}_{i}}\}$, where $p_{\mathcal{R}_{i}}=\sum_{\ell\in\mathcal{R}_{i}}b_{\ell}$.

\begin{remark}\label{rm1}
Building on real-life applications, scenarios such as environmental monitoring, or smart city implementations are ripe for MCS. In these contexts, tasks may show up uncertainly (e.g., identifying a city block or mapping). Given their unpredictable nature, there's a need to collect data efficiently. Moreover, the sensitivity of the worker's bid involved necessitates strong privacy mechanisms. It is in such scenarios that our algorithm works, providing not just incentive mechanisms but also ensuring the protection of worker bids via differential privacy. This becomes paramount to prevent inference attacks and protect the privacy of participating workers.
\end{remark}

\begin{example}
In this example, the platform has a task set $\mathcal{T}=\{\tau_{1},\tau_{2},\tau_{3},\tau_{4},\tau_{5}\}$ with five tasks and divides them into seven task subsets $\Gamma_{1}=\{\tau_{1},\tau_{2}\}$, $\Gamma_{2}=\{\tau_{2},\tau_{3}\}$, $\Gamma_{3}=\{\tau_{3},\tau_{1},\tau_{4}\}$, $\Gamma_{4}=\{\tau_{4},\tau_{5}\}$, $\Gamma_{5}=\{\tau_{4}\}$, $\Gamma_{6}=\{\tau_{2},\tau_{5}\}$ and $\Gamma_{7}=\{\tau_{2},\tau_{4},\tau_{5}\}$. There are seven workers whose costs are $c_{1}=1.4$, $c_{2}=1.8$, $c_{3}=2.8$, $c_{4}=2.6$, $c_{5}=3.1$, $c_{6}=3.3$ and $c_{7}=3.6$. As HERALD* is truthful, which will be demonstrated later, it follows that $b_i=c_i$. Before the arrival of the actually requested tasks, the platform first matches workers for all possible task subsets. We refer to this as the \textbf{task-worker matching phase} which is shown in Alg.~\ref{alg1}. For each task subset $\Gamma_{j} \in \mathcal{Y}$, the probability that worker $i$ with the unique bid $b_i$ is matched is $Pr(b_i)$. For each round of matching, the platform eventually assigns the task subset to the worker randomly according to the computed probability distribution. To prevent any monopolistic behavior and ensure the quality of the sensing task, we require that each task $\tau_{i}\in\mathcal{T}$ be matched by at least two different people. We assume that the matching set $\mathcal{P}=\{(\Gamma_{1},b_1),(\Gamma_{2},b_2),(\Gamma_{3},b_3),(\Gamma_{4},b_4),(\Gamma_{5},b_5),(\Gamma_{6},b_6),(\Gamma_{7},b_7)\}$.

We make the assumption that the task arrival follows a uniform distribution. If we set the number of tasks that arrive simultaneously to one, i.e., one task arrives at each time, then the task can be $\tau_{1}$, $\tau_{2}$, $\tau_{3}$, $\tau_{4}$ or $\tau_{5}$ with the same probability of $\frac{1}{5}$. Consequently, the selection threshold is $T=125.44$. The platform then executes the \textbf{winning selection phase} as described in Alg.~\ref{alg2}. In the first iteration, after computing the cost-effectiveness of all matching pairs, the condition in Line 6 of HERALD* is satisfied. Therefore, the platform performs a \texttt{type I selection} and selects matching pair $(\Gamma_{1},b_1)$ as the winning pair. In the second iteration, the condition in Line 6 still holds, and thus another \texttt{type I selection} is performed, selecting matching pair $(\Gamma_{4},b_4)$ as the winning pair. By repeating this process, we obtain that the HERALD* algorithm selects the final winning set $\mathcal{S}=\{(\Gamma_{1},b_1),(\Gamma_{2},b_2),(\Gamma_{4},b_4)\}$.

Next, the \textbf{payment determination phase} is carried out by the platform as shown in Alg.~\ref{alg3}. For instance, for worker $1$, whose covering set is $\mathcal{W}_{1}=\{2,3,6,7\}$, the replace set is $\mathcal{R}_{1}=\{2,3\}$. Thus, the payment for worker $1$ is calculated as $p_{1}=1.8+2.8=4.6$. Similarly, the payments for worker $2$ and worker $4$ are computed as $p_{2}=1.4+2.8=4.2$ and $p_{4}=3.6$, respectively. Moreover, when two tasks arrive simultaneously, they can be $\{\tau_{1},\tau_{1}\}$, $\{\tau_{2},\tau_{2}\}$, $\{\tau_{3},\tau_{3}\}$, $\{\tau_{4},\tau_{4}\}$, $\{\tau_{5},\tau_{5}\}$ with the same probability $\frac{1}{25}$, and $\{\tau_{1},\tau_{2}\}$, $\{\tau_{1},\tau_{3}\}$, $\{\tau_{1},\tau_{4}\}$,  $\{\tau_{1},\tau_{5}\}$, $\{\tau_{2},\tau_{3}\}$, $\{\tau_{2},\tau_{4}\}$, $\{\tau_{2},\tau_{5}\}$, $\{\tau_{3},\tau_{4}\}$, $\{\tau_{3},\tau_{5}\}$, $\{\tau_{4},\tau_{5}\}$ with the same probability $\frac{2}{25}$. The selection threshold is $T=181.248$. Subsequently, the HERALD* platform can execute the stages of \textbf{winning selection phase} and \textbf{payment determination phase} in a sequential manner in order to derive the set of winning entries and their respective payment amounts.
\end{example}

\begin{remark}\label{rm2}
It is apparent that when the input parameter $k$ (the number of tasks arriving simultaneously) is set to the total number of tasks $n$ in the platform's task set, HERALD* exhibits a probability of $\frac{A_n^n}{n^n}$, which is reduced to an offline incentive mechanism. Several prior studies have explored offline scenarios where the platform has complete knowledge of the task information. Hence, this implies that HERALD* can be employed in a wider range of scenarios than the conventional offline incentive mechanisms.
\end{remark}

\subsection{Design of Score Functions}
To apply the exponential mechanism to achieve the differential privacy of bidding results, it is necessary to devise score functions. Two score functions, namely a linear score function, and a logarithmic score function are created for this purpose. We will show that they have theoretical bounds on differential privacy and produce different impacts in simulations.

\textbf{Linear score function:}
$f_{lin}(x)=-x$. For any worker $i \in \mathcal{W}$, the probability that worker $i$ with bid $b_i$ is matched with any task subset $\Gamma_j \in \mathcal{Y}$ is

\begin{equation}\label{LIN score functuon}
\footnotesize
\begin{split}
Pr(b_i)\propto \left\{
             \begin{array}{lc}
             exp(-\frac{\epsilon b_i}{2 \Delta u b_{max}}),&\text{if $i \in \mathcal{W}$},\\
             {0},&\text{otherwise}.
             \end{array}
        \right.
\end{split}
\end{equation}

Since $u = f_{lin}(x) = -x$, we have
\begin{footnotesize}
\begin{equation}
\begin{split}
\Delta u&=max_{r \in R} max_{x,y:{||x-y||}_1\le 1} |u(x,r)-u(y,r)| \\
&= \frac{b_{max}-b_{min}}{b_{max}}. 
\end{split}
\end{equation}
\end{footnotesize}
In order to guarantee that the score function's value is non-negative, we apply the following normalization process.

\begin{footnotesize}
\begin{equation}\label{Normalize LIN score functuon}
\begin{split}
Pr(b_i)= \left\{
             \begin{array}{lc}
             \frac{exp(-\frac{\epsilon b_i}{2 (b_{max}-b_{min})})}{\sum_{j \in \mathcal{W}}exp(-\frac{\epsilon b_j}{2 (b_{max}-b_{min})})},&\text{if $i \in \mathcal{W}$},\\
             {0},&\text{otherwise}.
             \end{array}
        \right.
\end{split}
\end{equation}
\end{footnotesize}

\textbf{Logarithmic score function:}
$f_{\ln}(x)=-\ln(x)$. For any worker $i \in \mathcal{W}$, the probability that worker $i$ with bid $b_i$ is matched with any task subset $\Gamma_j \in \mathcal{Y}$ is
\begin{equation}\label{LN score functuon}
\footnotesize
\begin{split}
Pr(b_i)\propto \left\{
             \begin{array}{lc}
             exp(\frac{-\epsilon \ln\frac{b_i}{b_{max}}}{2 \Delta u}),&\text{if $i \in \mathcal{W}$},\\
             {0},&\text{otherwise}.
             \end{array}
        \right.
\end{split}
\end{equation}
Since $u = f_{\ln}(x)=-\ln(x)$, we have
\begin{equation}
\footnotesize
\begin{split}
\Delta u&=max_{r \in R}max_{x,y:{||x-y||}_1\le 1} 
|u(x,r)-u(y,r)| \\
&= \ln \frac{b_{max}}{b_{max}}-\ln \frac{b_{min}}{b_{max}} = -\ln\frac{1}{b_{max}} = \ln b_{max}, 
\end{split}
\end{equation}
where $b_{min}$ is normalized to $1$ and $b_{max}$ is a constant. We also need to normalize the score function to ensure that its value is non-negative
\begin{equation}\label{Normalize LN score functuon}
\footnotesize
\begin{split}
Pr(b_i)= \left\{
             \begin{array}{lc}
             \frac{exp\left(\frac{-\epsilon \ln \frac{b_i}{{b_{max}}} }{2 \ln b_{max}}\right)}{\sum_{j \in \mathcal{W}}exp\left(\frac{-\epsilon \ln \frac{b_j}{{b_{max}}} }{2 \ln b_{max}}\right)},&\text{if $i \in \mathcal{W}$},\\
             {0},&\text{otherwise}.
             \end{array}
        \right.
\end{split}
\end{equation}

\subsection{Analysis}
This subsection will provide evidence that HERALD* conforms to the characteristics outlined in Section \ref{subsecdes}.

\begin{theorem}\label{X-Btruth}
HERALD* is truthful.
\end{theorem}

The proof is given in Appendix A.

\begin{lemma}\label{X-Bir}
HERALD* is individually rational.
\end{lemma}
\begin{proof}
Theorem \ref{X-Btruth} demonstrates that each worker bids their actual cost $c_{i}$. The individual rationality of HERALD* is ensured by the fact that the payment made to each winner $i$ is equal to $p_{i}=p_{i}+\max\{b_{i},p_{\mathcal{R}_{i}}\}\geq b_{i}=c_{i}$.
\end{proof}

Besides proving truthfulness and individual rationality, we also demonstrate that HERALD* preserves the intended differential privacy of the bidding outcomes. Initially, we examine the impact of the linear score function on the differential privacy outcomes.

\begin{theorem}\label{linear DP}
For any constant $\epsilon > 0$, HERALD* with the linear score function preserves $\frac{\epsilon l}{2}$-differential privacy, where $\epsilon > 0$ is a constant and $l$ is the number of task subsets.
\end{theorem}

The proof is given in Appendix B. Next, we analyze the effect of the logarithmic score function on the differential privacy results.

\begin{theorem}\label{logarithmic DP}
For any constant $\epsilon > 0$, HERALD* with the logarithmic score function preserves $\frac{\epsilon l}{2}$-differential privacy, where $\epsilon > 0$ is a constant and $l$ is the number of task subsets.
\end{theorem}

The proof is given in Appendix C. In addition to its truthfulness, individual rationality, and differential privacy, HERALD* exhibits a low level of computational complexity, as can be observed.

\begin{proposition}\label{complexity-A}
The computational complexity of the HERALD* is $\mathcal{O}(lm + l^2 n^2 + ln^3)$.
\end{proposition}
\begin{proof}
In order to determine the computational complexity of HERALD*, we must analyze the task-worker matching phase (as outlined in Alg.~\ref{alg1}), the winning selection phase (in Alg.~\ref{alg2}), and the payment determination phase (as described in Alg.~\ref{alg3}), separately.

\begin{itemize}
  \item [1)] \emph{task-worker Matching Phase:} In Alg.~\ref{alg1}, the first loop (Lines 3-5) calculates the probability distribution over $m$ iterations with a key operation in Line 4 having $\mathcal{O}(m)$ complexity. The second loop (Lines 6-9) completes task-worker matching in $l$ iterations, where each matching in Line 7 also has $\mathcal{O}(m)$ complexity. Thus, the total computational complexity of this phase is $\mathcal{O}(m+lm)$, which simplifies to $\mathcal{O}(lm)$, accounting for both loops' combined computational requirements.

  
  \item [2)] \emph{Winning Selection Phase:} In Alg.~\ref{alg2}, the main loop (Lines 3--15) ends in up to $n$ rounds. Each round includes type I (Lines 7--8) and type II (Lines 9--10) selections, executed $l$ times, with complexity $\mathcal{O}(l)$ for minimum value identification. Lines 12-14, handling match search and modifications, contribute $\mathcal{O}(n^2)$ complexity. Overall, this phase's complexity is $\mathcal{O}(n l^2 + n^3)$, considering the detailed computations required.
  

  \item [3)] \emph{Payment Determination Phase:} In Alg.~\ref{alg3}, the main loop (Lines 1--9) completes in at most $l$ iterations. Each iteration includes constructing a covering set with complexity $\mathcal{O}(ln^2)$ (Line 3) and forming a replaced set over $n$ iterations (Lines 5--9), contributing $\mathcal{O}(nl + n^3)$ in complexity. Therefore, the overall computational complexity of Alg.~\ref{alg3} is $\mathcal{O}(n^2l^2 + ln^3)$, reflecting the combined computational requirements of constructing covering and replace sets.
  

\end{itemize}
By merging the task-worker matching phase, winning selection phase, and payment determination phase, HERALD* has a computational complexity of $\mathcal{O}(lm + l^2 n^2 + ln^3)$. We have specified in Section \ref{subsecsys} that $l < mn$. Some empirical running times for the three algorithms are documented in Appendix F for further reference.
\end{proof}

In the following parts, we will demonstrate the competitive ratio of expected social cost attained by HERALD* when the tasks in $\mathcal{T}$ arrive according to a uniform distribution. Specifically, we let $T=64\mathbb{E}[C_{\mathcal{OPT}}(\mathcal{A},W)]$, where $\mathcal{A}\subseteq\mathcal{T}$ is a subset of $k$ tasks that could potentially arrive concurrently from the task set $\mathcal{T}$. In order to calculate the competitive ratio of the expected social cost of HERALD*, we analyze the costs of type I selection and type II selection independently.

\begin{lemma}\label{cost_type1}
HERALD* achieves a competitive ratio of $\mathcal{O}(\ln n)$ on the expected social cost through type I selection when the tasks in the task set $\mathcal{T}$ are distributed uniformly.
\end{lemma}

The proof is given in Appendix D. We still need to establish a bound on the expected social cost of workers in the winning set $\mathcal{S}$ chosen by HERALD* via type II selection. To accomplish this, we must introduce the following notations.

We define $\mathcal{S}_{II}=\{1,\ldots,\ell\}$ as the group of workers in the winning set $\mathcal{S}$ selected by HERALD* via type II selection in the specified sequence. We also define $\widetilde{\mathcal{T}}_{i}$ as the set of tasks whose sensory data is not gathered right before worker $i$ is selected. Additionally, $n_{i}=|\widetilde{\mathcal{T}}_{i}|$ represents the number of tasks in $\widetilde{\mathcal{T}}_{i}$, and $k_{i}=n_{i}\frac{k}{n}$ is the expected count of requested tasks arriving from $\widetilde{\mathcal{T}}_{i}$. We use $\mathcal{A}_{i}$ to denote the subset of $\mathcal{A}$ that contains requested tasks belonging only to $\widetilde{\mathcal{T}}_{i}$. Furthermore, we define $\mathcal{S}^{*}(\mathcal{A},W)$ as the winning set with the least social cost for any $\mathcal{A}$ set. Then, let $\mathcal{S}^{\prime}(\mathcal{A}_{i},W)$ be the subset of $\mathcal{S}^{*}(\mathcal{A},W)$ such that for each task $\tau_{j}\in\mathcal{A}_{i}$, the worker in $\mathcal{S}^{\prime}(\mathcal{A}_{i},W)$ has the corresponding task subset containing task $\tau_{j}$ and has the least cost among workers in $\mathcal{S}^{*}(\mathcal{A},W)$.

\begin{lemma}\label{cost_type2}
When the arrivals of tasks in the task set $\mathcal{T}$ follow a uniform distribution, HERALD* achieves a competitive ratio of $\mathcal{O}(\ln ln)$ on the expected social cost through type II selection.
\end{lemma}

The proof is given in Appendix E. By combining Lemma \ref{cost_type1} and Lemma \ref{cost_type2}, the following theorem is established.
\begin{theorem}\label{expected_ratio}
HERALD* achieves a competitive ratio of $\mathcal{O}(\ln ln)$ on expected social cost when the arrivals of tasks in the task set $\mathcal{T}$ follow a uniform distribution.
\end{theorem}

Based on the results from Theorem~\ref{expected_ratio}, we can deduce that HERALD* yields a low expected social cost. Therefore, HERALD* can be applied to many scenarios with uncertain sensing tasks.

\section{Performance Evaluation}
The subsequent section will showcase the benchmark methods utilized for evaluating the performance of HERALD* and expound on the simulation settings employed in the experiment. Furthermore, the outcomes of the simulation will be presented.

\subsection{Baseline Methods}

As very few prior works relate directly to uncertain tasks in the MCS system, we chose to repurpose mechanisms from related optimization problems, such as set covering, to evaluate HERALD*. We compared HERALD* with CONE and COSY in a simulation, all adapted from the set covering problems in \cite{grandoni2013set}. Despite being part of our original incentive mechanism, CONE and COSY's links to optimization problems make them suitable baseline mechanisms.

\emph{COst-effectiveNEss greedy auction (CONE):}
When it comes to uncertain tasks without real-time constraints, the platform is aware that the tasks belonging to set $\mathcal{T}$ are expected to arrive in the future with a probability distribution. To obtain sensory data for these tasks, the platform computes the CF for each matching pair and picks the winning pair, denoted by $(\Gamma_{j}, b_i)$, which has the lowest CF $\frac{b_{i}}{|\Gamma_{j}\cap\mathcal{T}|}$ value among those of matching pairs in each iteration. The platform then proceeds to acquire the sensory data from worker $i$.

\emph{COSt greedY auction (COSY):}
When it comes to uncertain tasks without real-time constraints, the platform acquires sensory data by examining the bids of matching pairs and determining the winning pair, denoted by $(\Gamma_{j}, b_i)$. This pair has the lowest bid value $b_{i}$ among those of winning pairs and corresponds to a task subset $\Gamma_{j}$ that includes at least one task that has not yet been covered in the current iteration. After selecting the winning pair, the platform proceeds to collect the sensory data from worker $i$.

The payment determination process in both CONE and COSY is identical to that of HERALD*. It is evident that both CONE and COSY are truthful and individually rational.

\subsection{Simulation Settings}
Table \ref{table_1} outlines key metrics for various scenarios, including the cost $c_{i}$ that worker $i$ performs the assigned subset $\Gamma_{j}$ of matching tasks, along with the number of tasks $|\Gamma_{j}|$ in each subset, the number of workers $m$ and the number of sensing tasks $n$.


We present an assessment of HERALD*, which highlights how the expected social cost and expected total payment are affected by the number of workers and sensing tasks. Specifically, we conduct two evaluations: the first (setting I) fixes the number of sensing tasks at $n = 120$ and varies the number $m$ of workers from $60$ to $150$ in steps of $5$. The second (setting II) maintains the number of workers at $m = 80$ while incrementally varying the number $n$ of sensing tasks from $80$ to $160$ in steps of $5$. In both scenarios, we randomly and independently sample the cost $c_{i}$ of worker $i$ and the number of tasks in task subset $\Gamma_{j}$ from uniform distributions within the intervals $[1,5]$ and $[15,20]$, respectively.


We will now examine how worker costs affect the expected social cost and expected total payment yielded by HERALD*. To investigate the impact of worker costs $c_{i}$, we consider three different intervals: $[1,5]$, $[5,10]$ and $[10,15]$, in setting III. In this setting, we also randomly sample the number $|\Gamma_{j}|$ of tasks in the subset from the interval $[15,20]$. Additionally, we set the number $n$ of sensing tasks to $120$, while varying the number $m$ of workers from $60$ to $150$.

In setting IV, we assess how the number of tasks in subsets influences HERALD*'s social cost and total payment. Specifically, we examine three task intervals in setting IV: $[10,15]$, $[15,20]$, and $[20,25]$. Workers' costs are sampled from $[1,5]$. We fix the sensing tasks at $n = 120$ and vary worker numbers $m$ from $60$ to $150$. Across all settings, we maintain a differential privacy coefficient of $\epsilon = 0.1$.

We also investigate how the privacy parameter $\epsilon$ affects HERALD*'s social cost and total payment. Due to space limitations, please see Appendix G for details.


\begin{table}[ht]
\caption{Simulation Settings for HERALD*.}\label{table_1}
\begin{tabular}{|l|l|l|l|l|}
\hline
Settings &\makecell[l]{individual \\ cost $c_{i}$}& \makecell[l]{number $|\Gamma_{j}|$ of \\ each task subset} & \makecell[l]{number $m$ \\ of workers} & \makecell[l]{number $n$ \\ of sensing \\ tasks}\\
\hline
I & $[1,5]$ & $[15,20]$& $[60,150]$& $120$ \\
\hline
II & $[1,5]$ & $[15,20]$& $80$& $[80,160]$ \\
\hline
III & \makecell[l]{$[1,5],[5,10]$,\\ $[10,15]$} &$[15,20]$& $[60,150]$& $120$ \\
\hline
IV & $[1,5]$ & \makecell[l]{$[10,15],[15,20]$,\\ $[20,25]$}& $[60,150]$& $120$ \\
\hline
\end{tabular}
\vspace{-0.2in}
\end{table}

\begin{figure}[ht]
\centering
\begin{minipage}{0.85\linewidth}
\centering
 \includegraphics[width=0.85\linewidth]{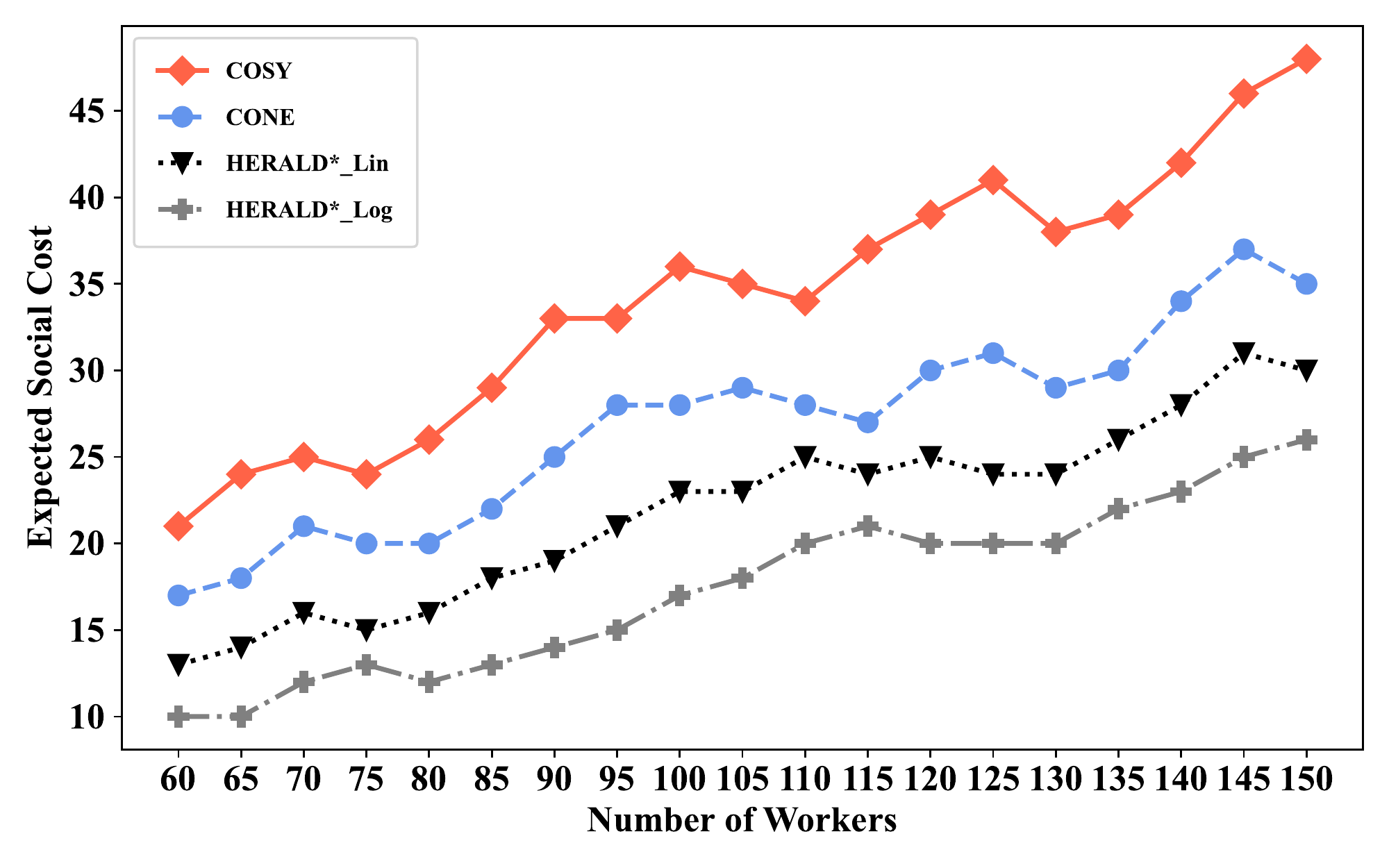}
 \centerline{\footnotesize{\quad (a)}}
 \end{minipage}%
 \qquad
\centering
\begin{minipage}{0.85\linewidth}
\centering
\includegraphics[width=0.85\linewidth]{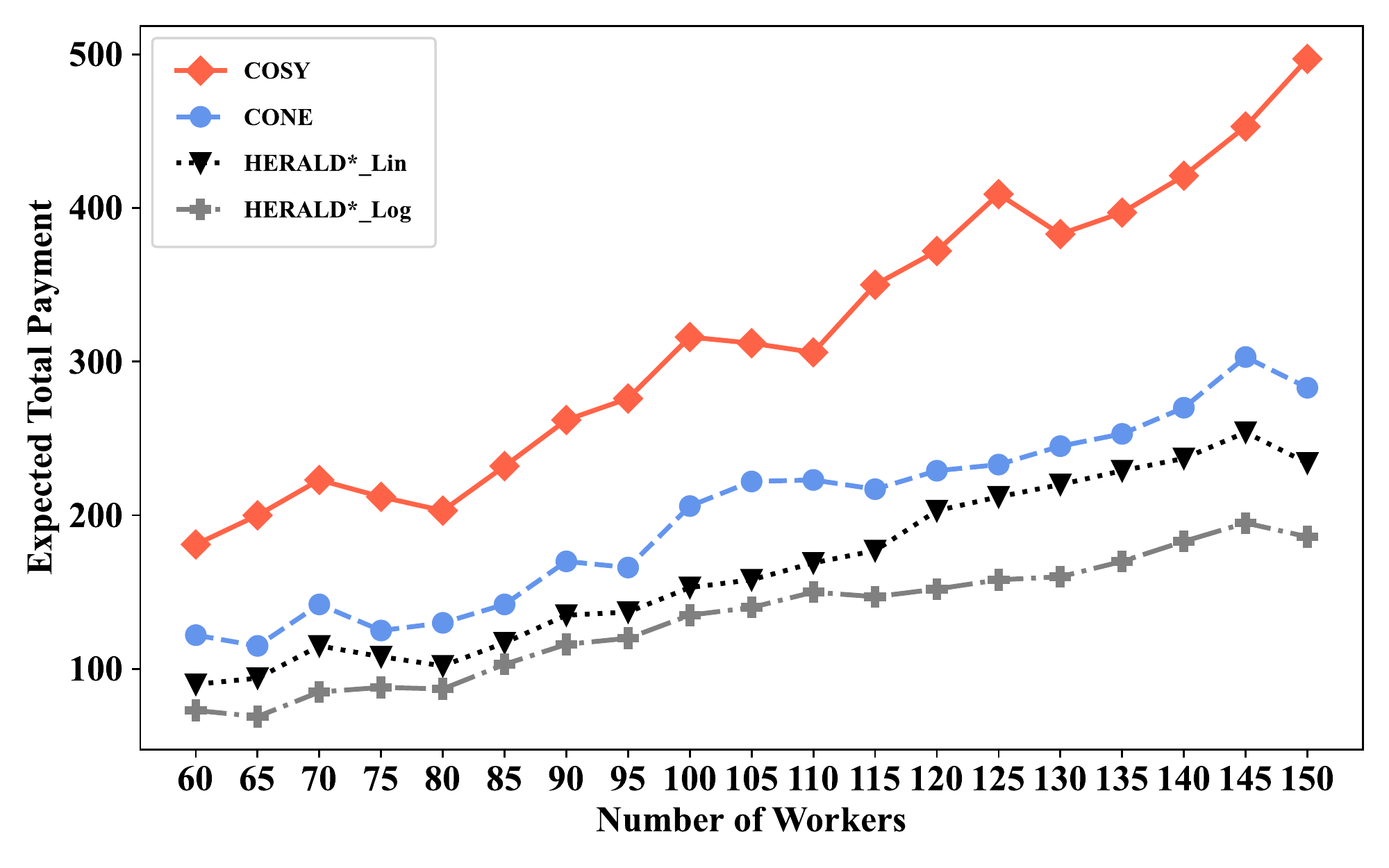}
 \centerline{ \footnotesize{\quad (b)}}
\end{minipage}%
\caption{(a). Expected social cost versus different numbers of workers for uncertain tasks. (b). Expected total payment versus different numbers of workers for uncertain tasks.}\label{Setting1}
\vspace{-0.2in}
\end{figure}

\begin{figure}[ht]
\centering
\begin{minipage}{0.85\linewidth}
\centering
 \includegraphics[width=0.85\linewidth]{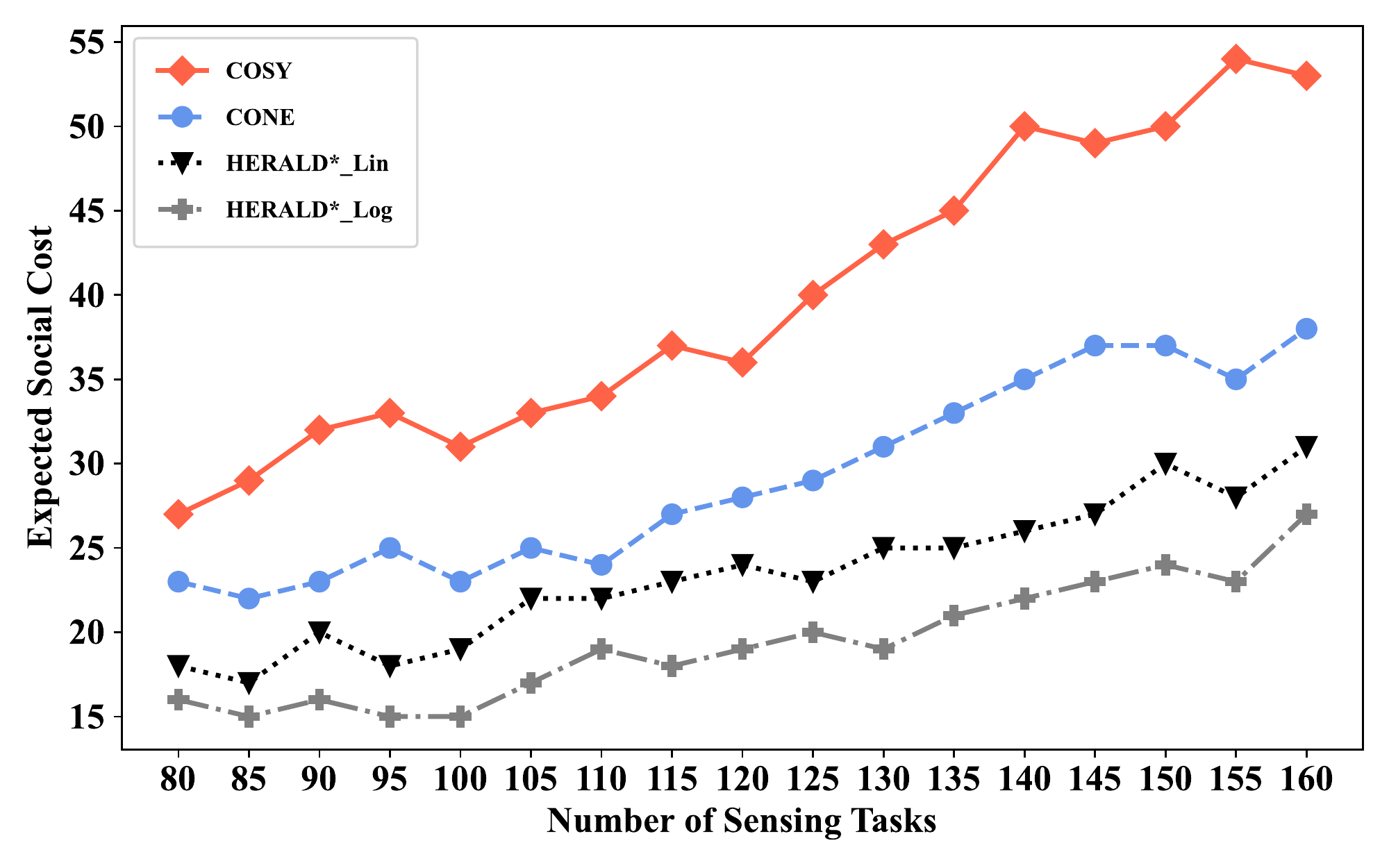}
 \centerline{\footnotesize{\quad (a)}}
 \end{minipage}%
 \qquad
\centering
\begin{minipage}{0.85\linewidth}
\centering
\includegraphics[width=0.85\linewidth]{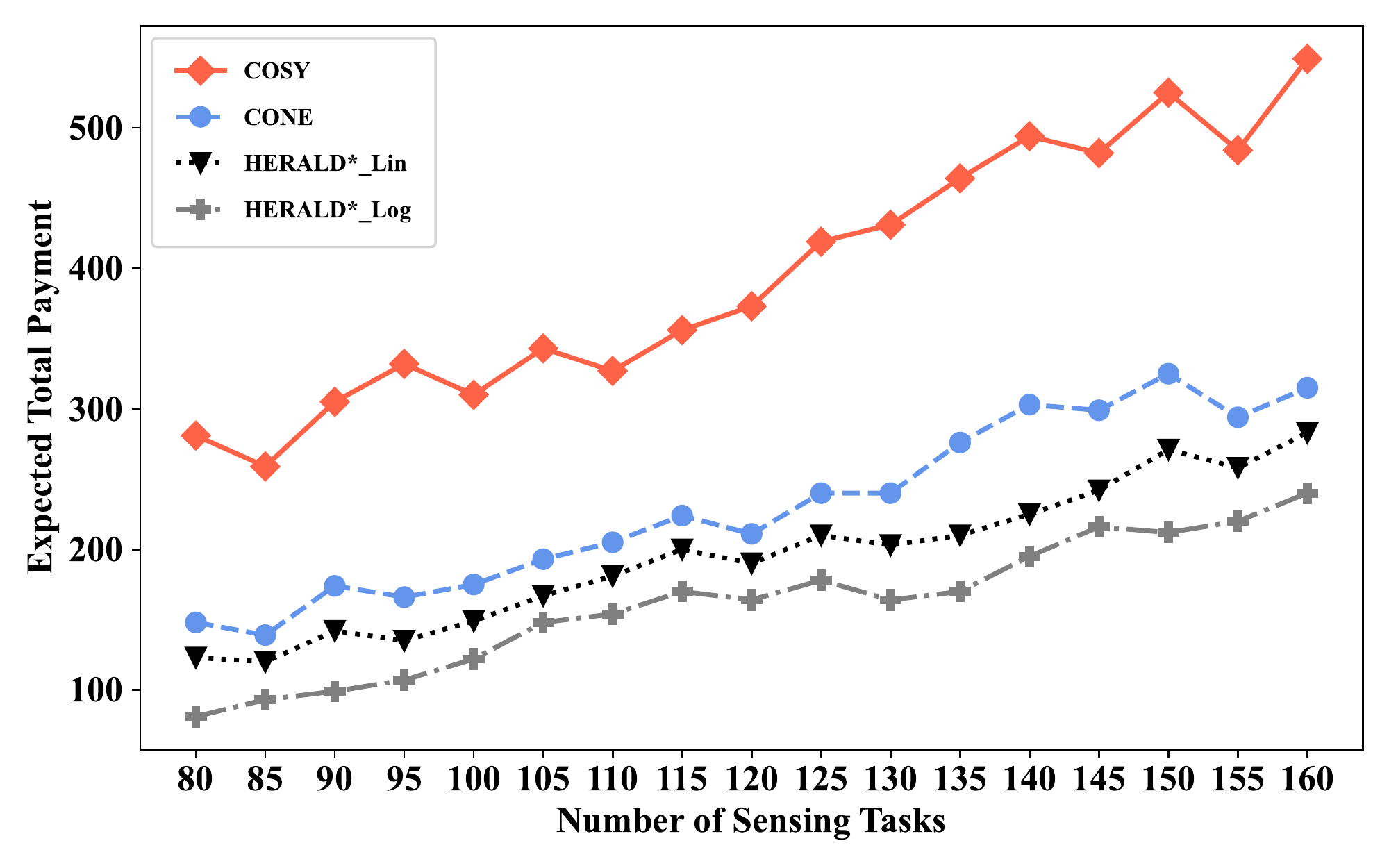}
 \centerline{ \footnotesize{\quad (b)}}
\end{minipage}%
\caption{(a). Expected social cost versus different numbers of sensing tasks for uncertain tasks. (b). Expected total payment versus different numbers of sensing tasks for uncertain tasks.}\label{Setting2}
\vspace{-0.2in}
\end{figure}

\begin{figure}[ht]
\centering
\begin{minipage}{0.85\linewidth}
\centering
 \includegraphics[width=0.85\linewidth]{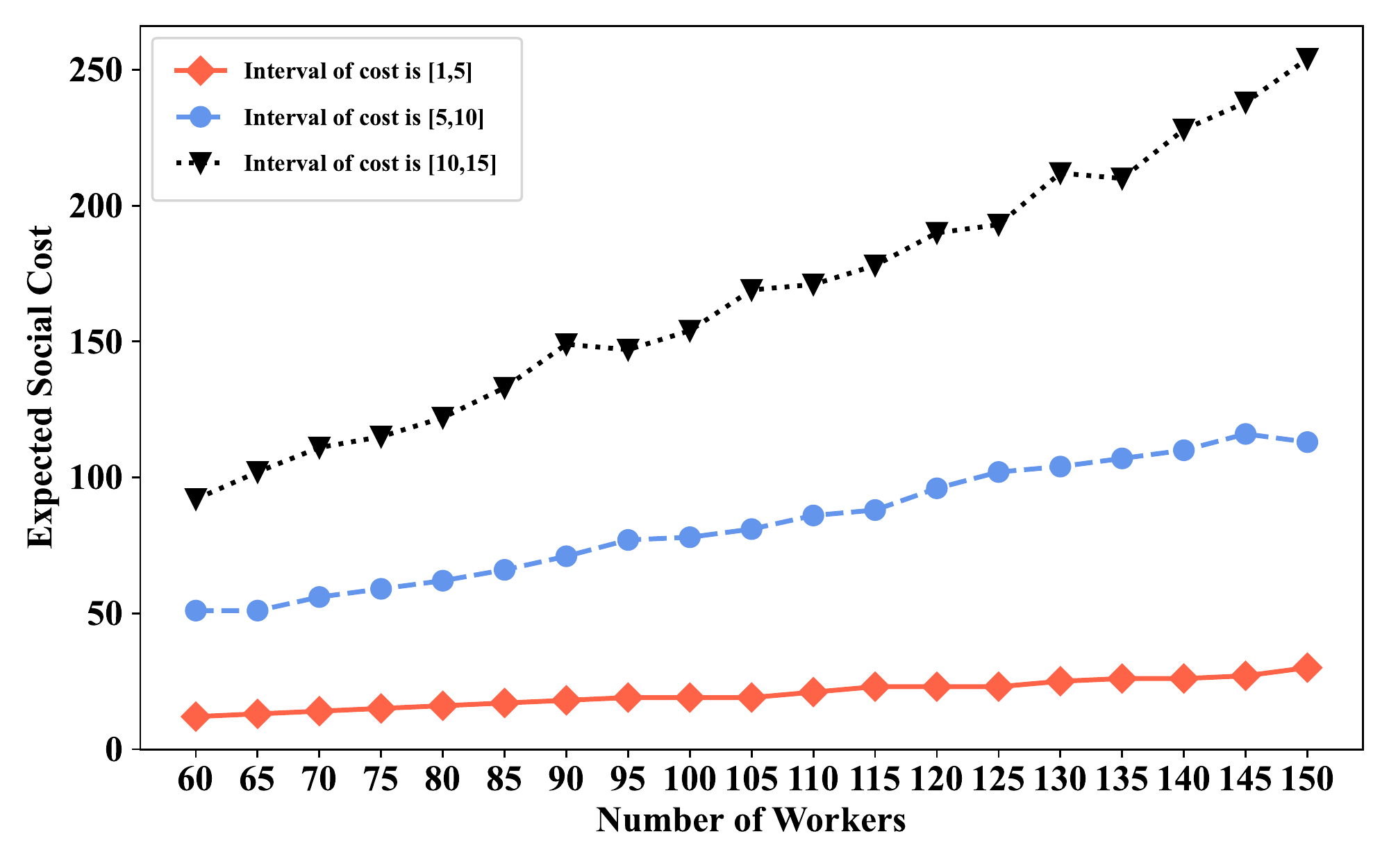}
 \centerline{\footnotesize{\quad (a)}}
 \end{minipage}%
 \qquad
\centering
\begin{minipage}{0.85\linewidth}
\centering
\includegraphics[width=0.85\linewidth]{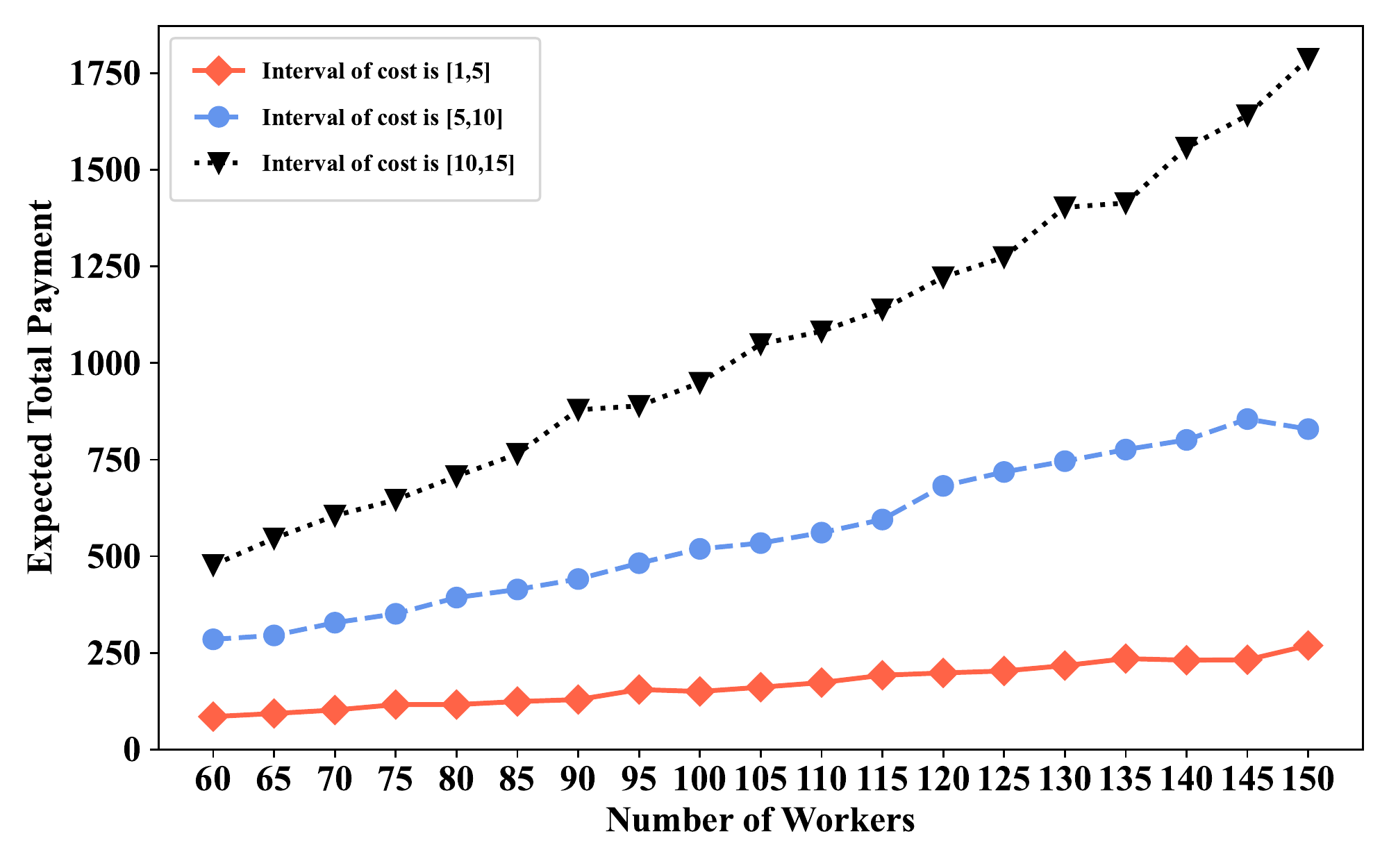}
 \centerline{ \footnotesize{\quad (b)}}
\end{minipage}%
\caption{(a). The impact of worker's cost on the expected social cost obtained by HERALD* for uncertain tasks  with the liner score function. (b). The impact of worker's cost on the expected total payment obtained by HERALD* for uncertain tasks  with the liner score function.}\label{Lin_Setting3}
\vspace{-0.2in}
\end{figure}

\begin{figure}[ht]
\centering
\begin{minipage}{0.85\linewidth}
\centering
 \includegraphics[width=0.85\linewidth]{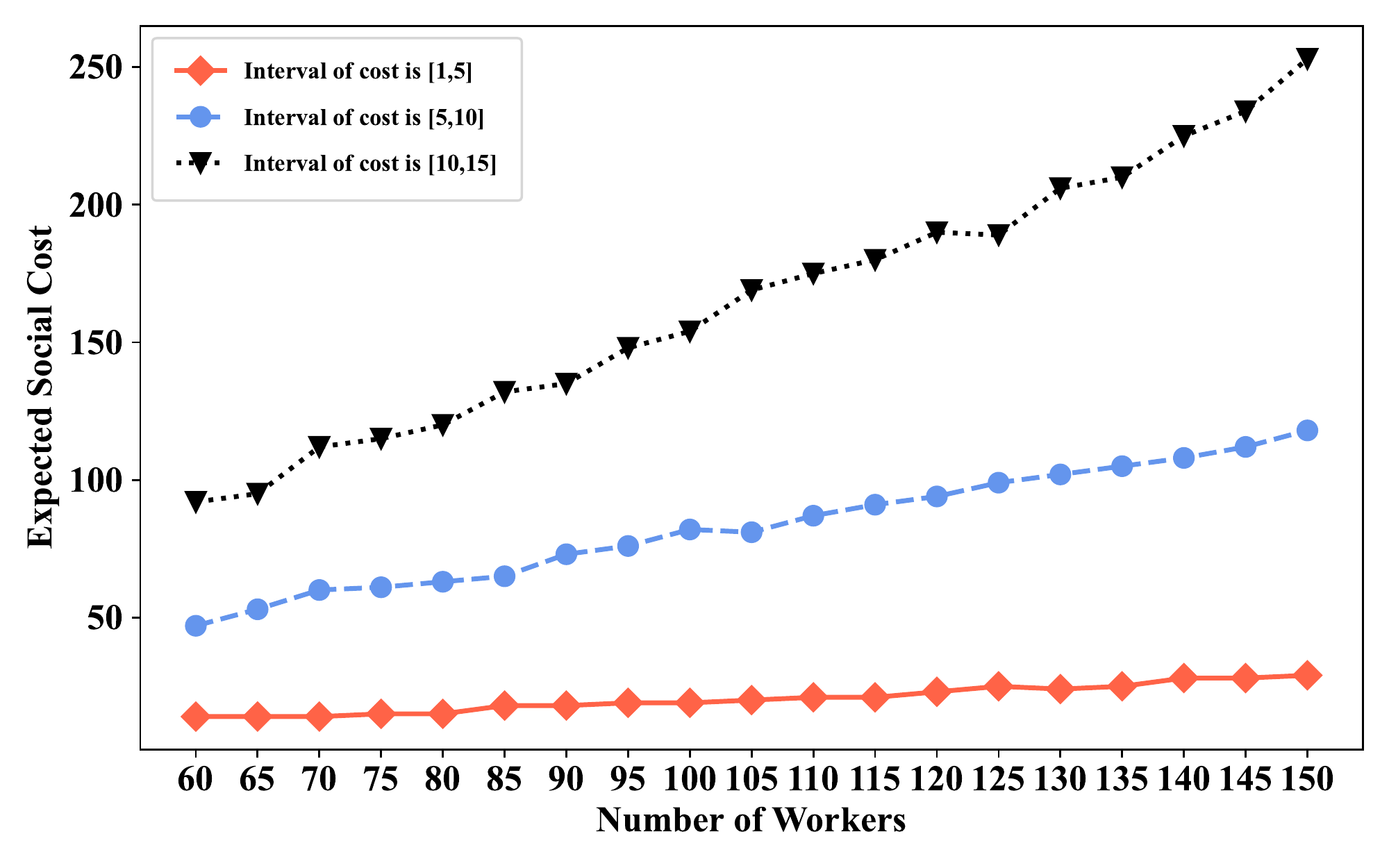}
 \centerline{\footnotesize{\quad (a)}}
 \end{minipage}%
 \qquad
\centering
\begin{minipage}{0.85\linewidth}
\centering
\includegraphics[width=0.85\linewidth]{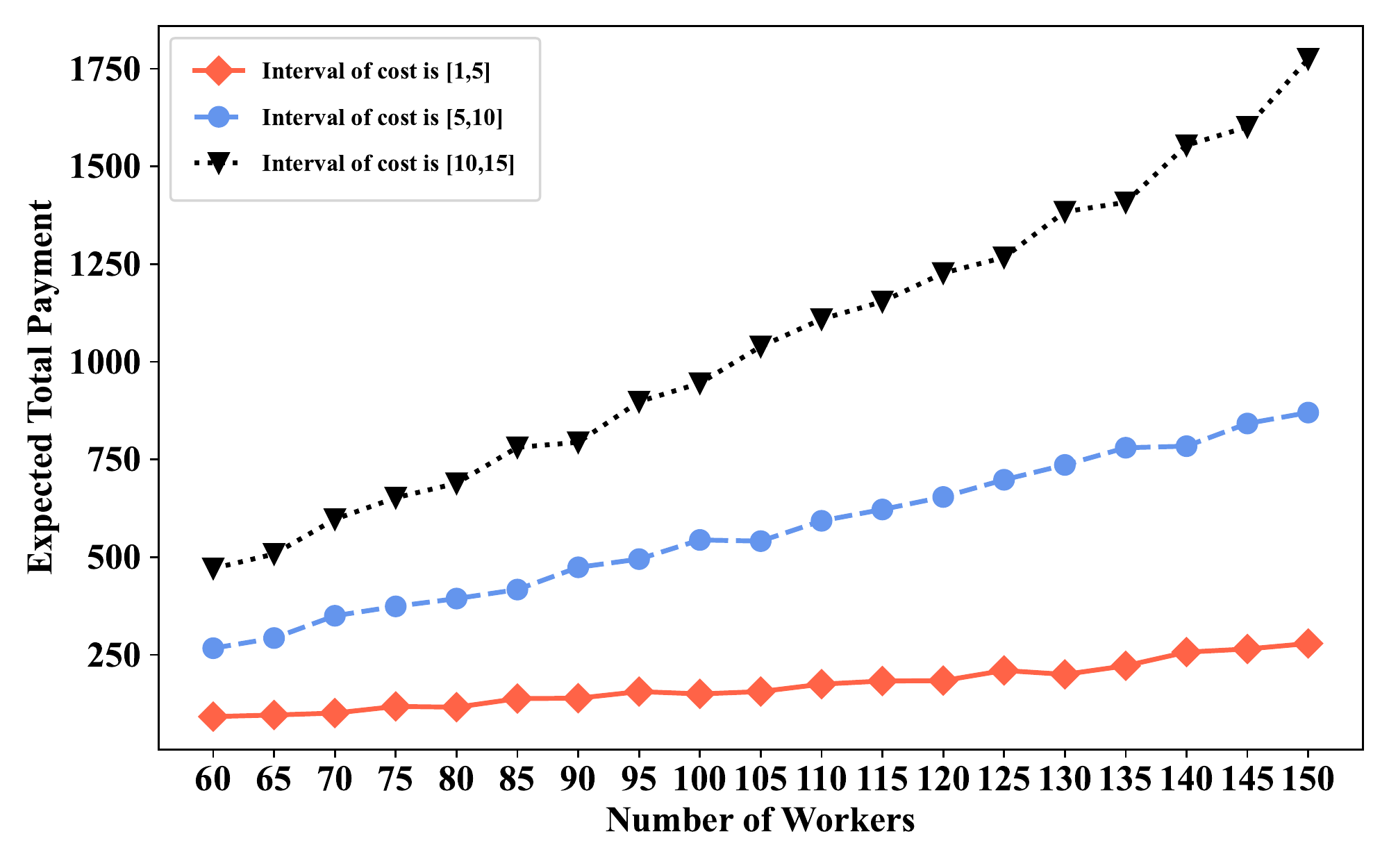}
 \centerline{ \footnotesize{\quad (b)}}
\end{minipage}%
\caption{(a). The impact of worker's cost on the expected social cost obtained by HERALD* for uncertain tasks with the logarithmic score function. (b). The impact of worker's cost on the expected total payment obtained by HERALD* for uncertain tasks with the logarithmic score function.}\label{Log_Setting3}
\vspace{-0.2in}
\end{figure}

\begin{figure}[ht]
\centering
\begin{minipage}{0.85\linewidth}
\centering
 \includegraphics[width=0.85\linewidth]{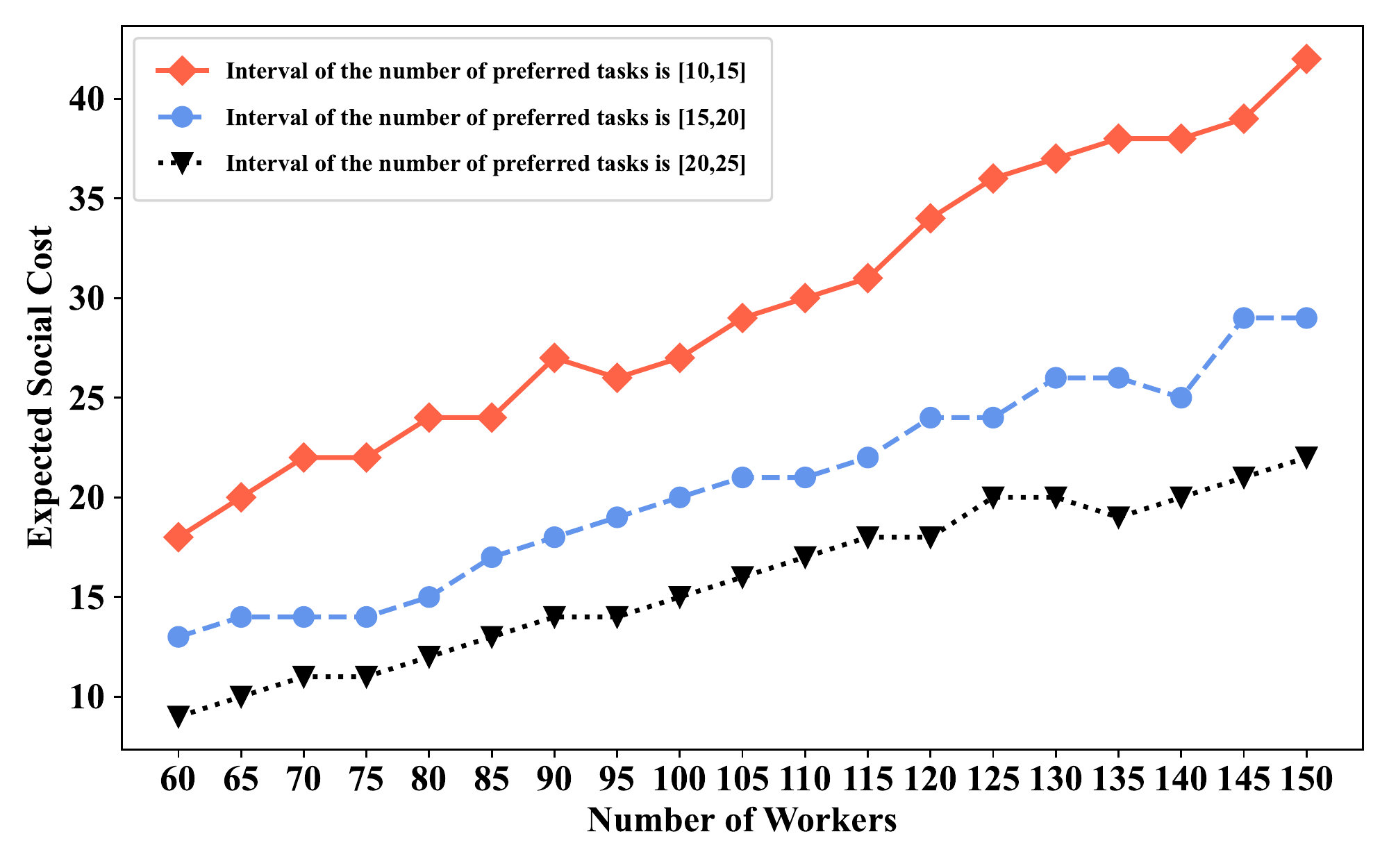}
 \centerline{\footnotesize{\quad (a)}}
 \end{minipage}%
 \qquad
\centering
\begin{minipage}{0.85\linewidth}
\centering
\includegraphics[width=0.85\linewidth]{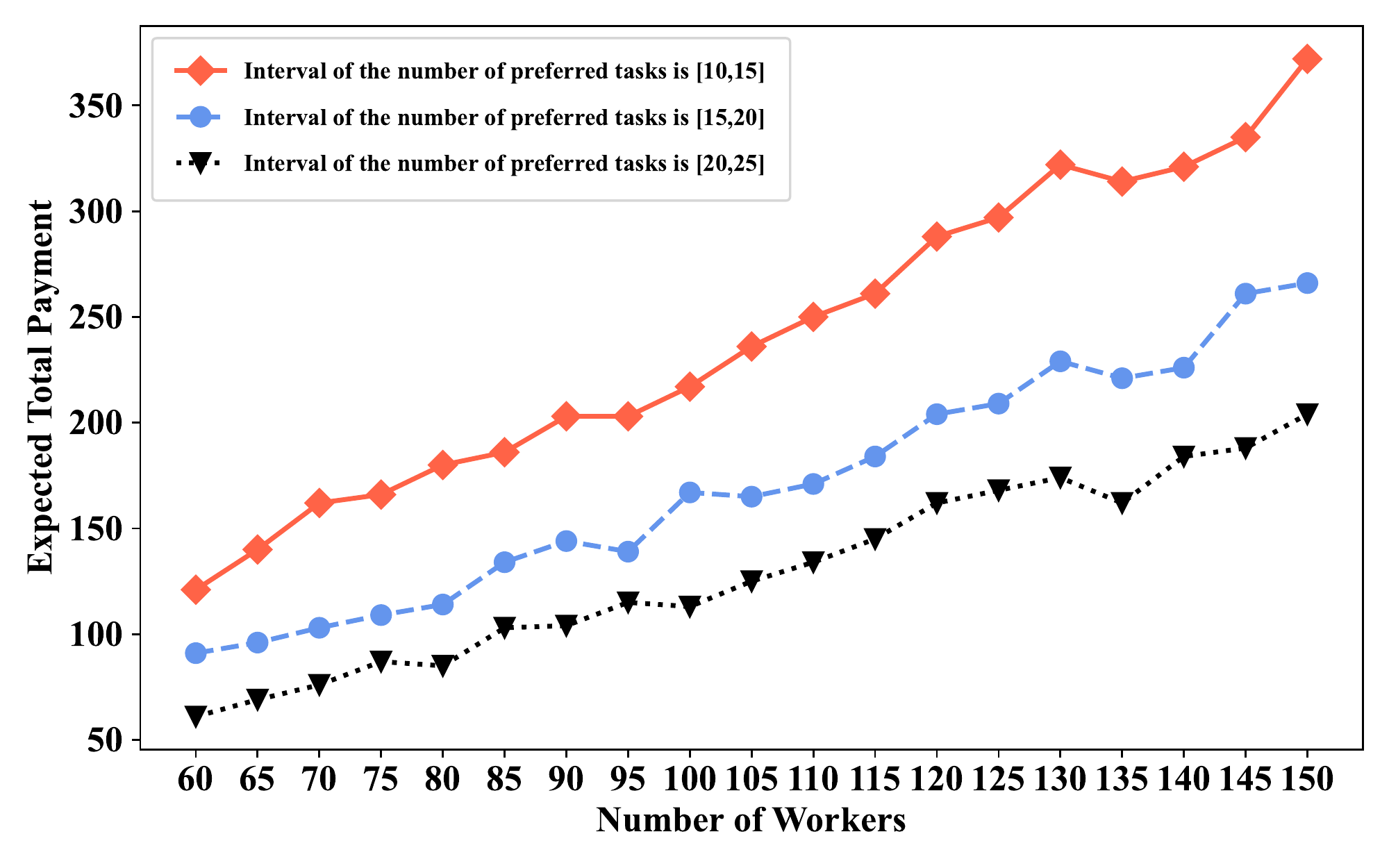}
 \centerline{ \footnotesize{\quad (b)}}
\end{minipage}%
\caption{(a). The impact of the number of worker's matching tasks on the expected social cost obtained by HERALD* for uncertain tasks with the liner score function. (b). The impact of the number of worker's matching tasks on the expected total payment obtained by HERALD* for uncertain tasks with the liner score function.}\label{Lin_Setting4}
\vspace{-0.2in}
\end{figure}

\begin{figure}[ht]
\centering
\begin{minipage}{0.85\linewidth}
\centering
 \includegraphics[width=0.85\linewidth]{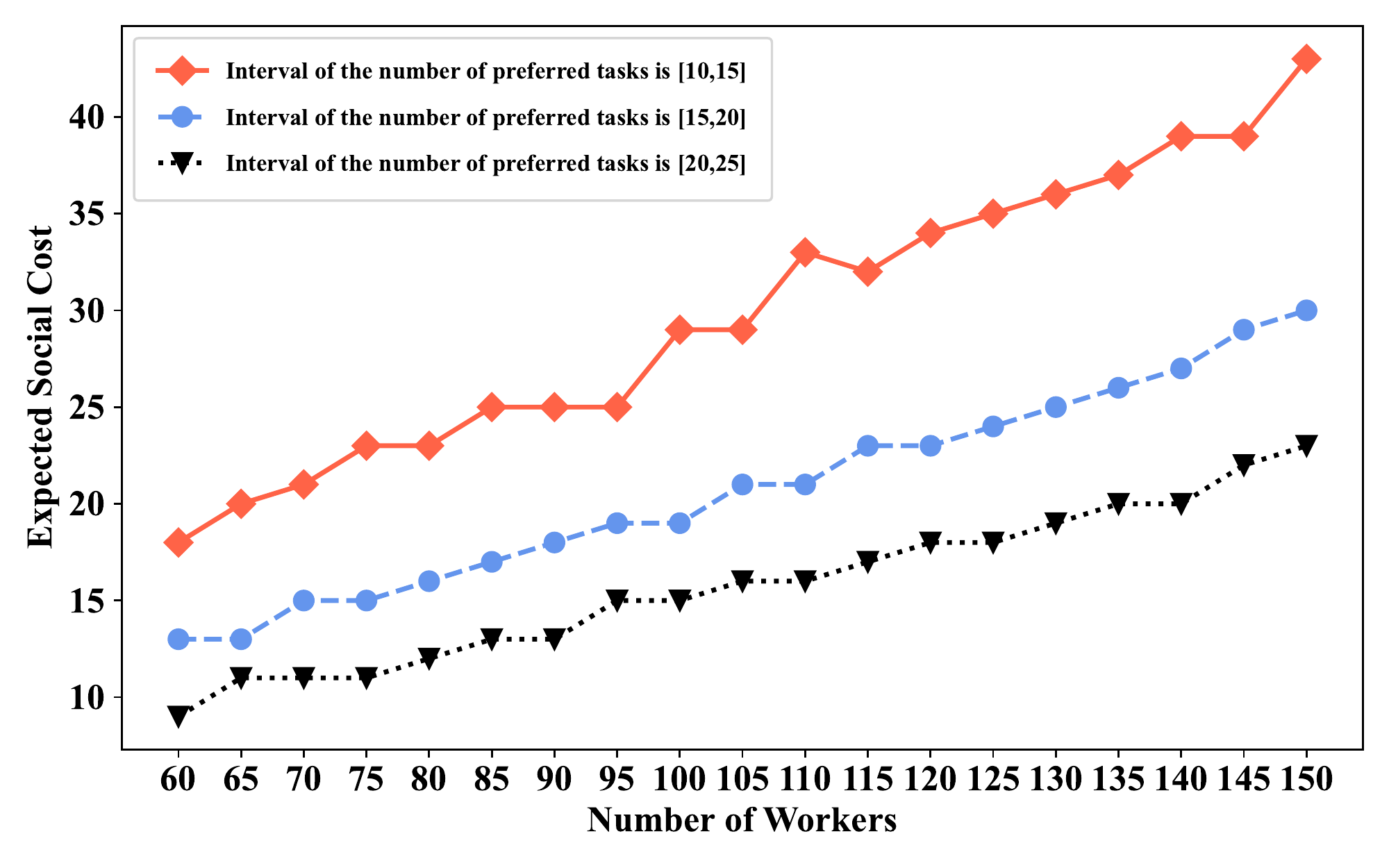}
 \centerline{\footnotesize{\quad (a)}}
 \end{minipage}%
 \qquad
\centering
\begin{minipage}{0.85\linewidth}
\centering
\includegraphics[width=0.85\linewidth]{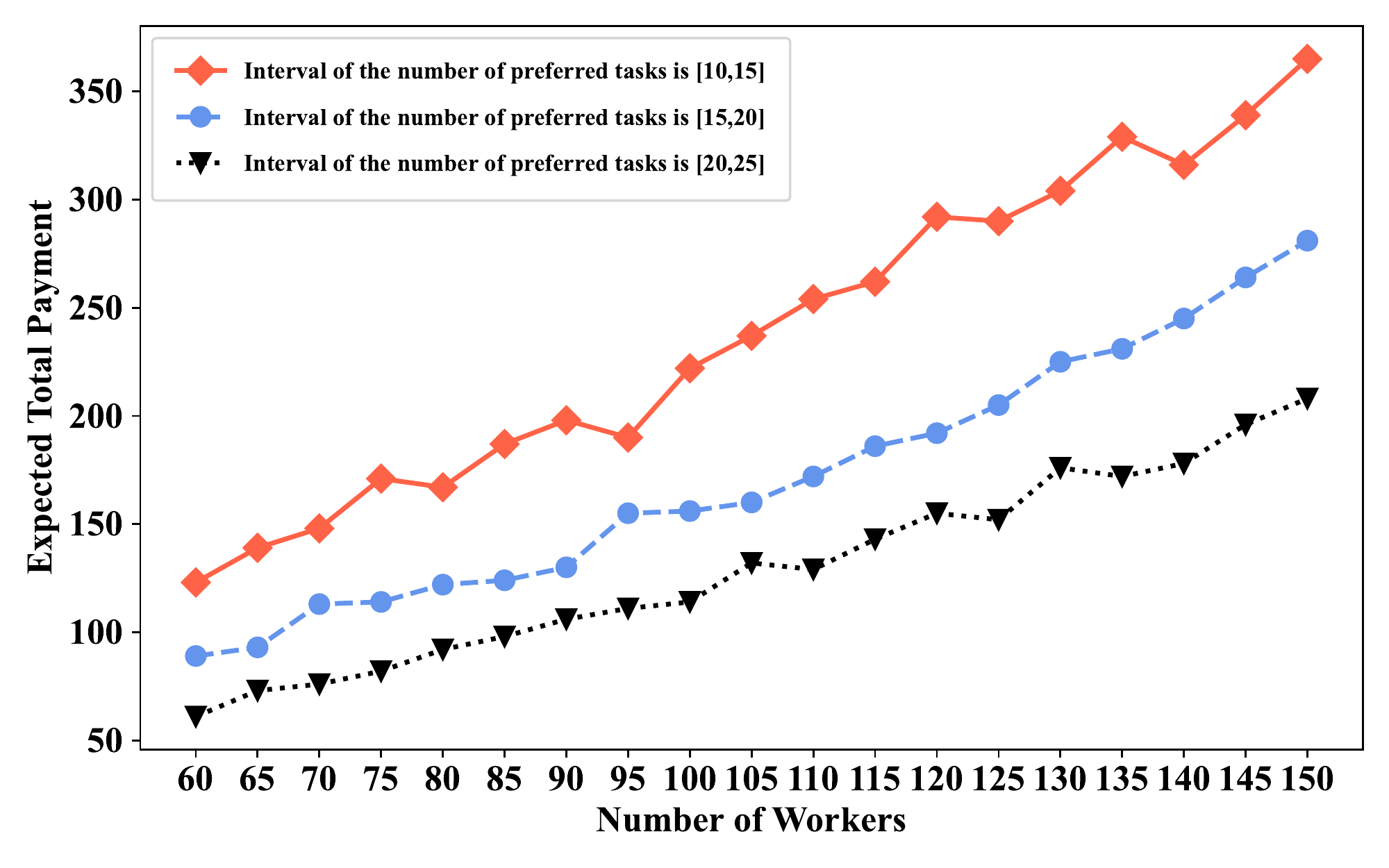}
 \centerline{ \footnotesize{\quad (b)}}
\end{minipage}%
\caption{(a). The impact of the number of worker's matching tasks on the expected social cost obtained by HERALD* for uncertain tasks with the logarithmic score function. (b). The impact of the number of worker's matching tasks on the expected total payment obtained by HERALD* for uncertain tasks with the logarithmic score function.}\label{Log_Setting4}
\vspace{-0.2in}
\end{figure}

\subsection{Simulation Results}
In Fig.~\ref{Setting1}, we analyze the effect of the number of workers. Specifically, Fig.~\ref{Setting1}(a) and Fig.~\ref{Setting1}(b) illustrate the impact on the expected social cost and expected total payment obtained by HERALD*, with results averaged over 100 runs. It is observed that HERALD* performs better than CONE and COSY. Interestingly, the expected social cost and expected total payment obtained by HERALD* using the logarithmic score function is lower than those obtained using the linear score function. This is due to the logarithmic score function giving higher chances of selection to users with low bids, resulting in a preference for such users.

Fig.~\ref{Setting2} examines the impact of the number of tasks on HERALD*'s performance, with results averaged over 100 runs. As depicted in Fig.~\ref{Setting2}(a) and Fig.~\ref{Setting2}(b), which show the expected social cost and expected total payment, respectively. Similar to the findings in Fig.~\ref{Setting1}, HERALD* outperforms CONE and COSY in this setting. Additionally, as the number of tasks increases, the expected social cost and expected total payment of HERALD* also increase due to the need for more workers to collect sensory data. Furthermore, consistent with the earlier results, the expected social cost and expected total payment of HERALD* with the logarithmic score function are lower than those of HERALD* with the linear score function for the same reasons.

Fig.~\ref{Lin_Setting3} depicts the effect of workers' cost on the performance of HERALD* under the linear score function. Specifically, Fig.~\ref{Lin_Setting3}(a) and Fig.~\ref{Lin_Setting3}(b) illustrate the impact of workers' cost on the expected social cost and expected total payment generated by HERALD*. The results indicate that as the workers' cost increases, both the expected social cost and expected total payment of HERALD* also increase. This is because a higher workers' cost implies that more social cost is required for the same tasks, and the platform has to pay more to the workers compared to the scenario with a lower workers' cost. Similar findings are observed for HERALD* under the logarithmic score function, as shown in Fig.~\ref{Log_Setting3}.

Fig.~\ref{Lin_Setting4} displays the influence of the number of matching tasks per worker on HERALD*'s expected social cost and expected total payment under the linear score function. Fig.~\ref{Lin_Setting4} (a) and \ref{Lin_Setting4}(b) highlight this impact, showing a decline in both the expected social cost and expected total payment as the number of tasks per worker increases. This trend is attributed to the reduced need for workers as task allocation per worker rises, leading to lower costs and payments in HERALD*. A similar trend is observed under the logarithmic score function, as depicted in Fig.~\ref{Log_Setting4}.

\section{Conclusion and Future Work}
In this manuscript, we introduce HERALD*, a novel incentive mechanism for a task allocation system in which tasks without real-time constraints arrive randomly according to a probability distribution. Our investigation indicates that HERALD* meets several desirable properties, such as truthfulness, individual rationality, differential privacy, low computational complexity, and low social cost. More specifically, we have shown that HERALD* guarantees $\frac{\epsilon l}{2}$-differential privacy for both linear and logarithmic score functions, and it achieves a competitive ratio of $\ln ln$ on expected social cost. We have also validated the effectiveness of HERALD* through both theoretical analysis and extensive simulations.

Moving forward, we plan to tackle the crucial privacy concerns associated with storing sensory data. Our future endeavors include integrating HERALD* with advanced data privacy algorithms (such as \cite{sun2021two,wang2022triple,wang2023differentially}). This integration aims to elevate our platform's data protection standards. Additionally, we plan to tackle ``response fatigue'' to improve user experience and system usability. We will introduce user-customizable interaction levels, intelligent notifications, and extensive user support to reduce decision overload and enhance the practicality of HERALD* in real-world applications.

%
%
%
%


\vspace{-33pt}


\begin{IEEEbiography}[{\includegraphics[width=1in,height=1.25in,clip,keepaspectratio]{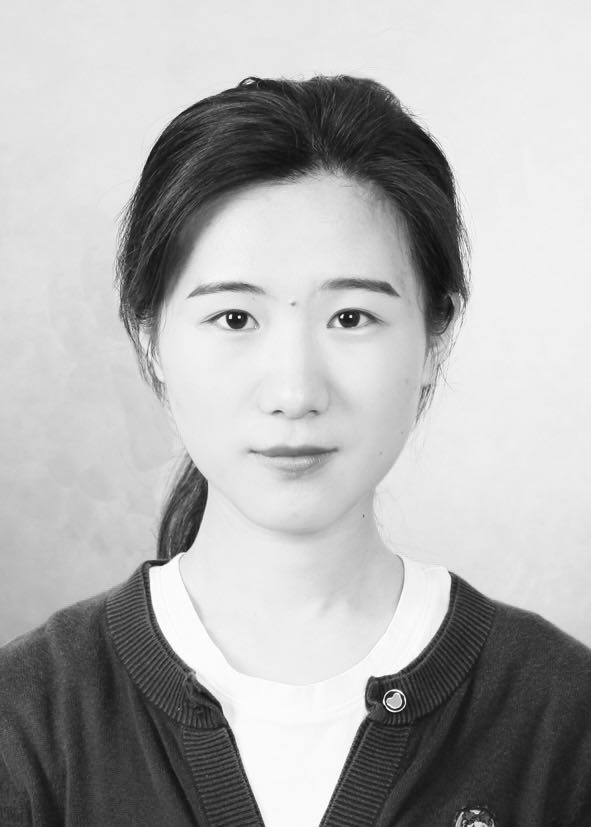}}]{Xikun Jiang} received her Ph.D. degree from the Department of Computer Science and Engineering, Shanghai Jiao Tong University, China in 2023. Presently, she holds a postdoctoral position in the Department of Computer Science (DIKU) at the University of Copenhagen (UCPH), Denmark. Her research interests include mobile crowdsensing, online markets, and machine learning. 
\end{IEEEbiography}

\vspace{-33pt}
\begin{IEEEbiography}[{\includegraphics[width=1in,height=1.25in,clip,keepaspectratio]{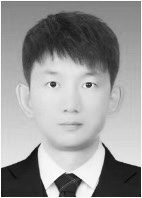}}]{Chenhao Ying} received Ph.D. degree from Department of Computer Science and Engineering, Shanghai Jiao Tong University, China in 2022. Before that, he received the B.E. degree from the Department of Communication Engineering, Xidian University, China, in 2016. He is a research assistant professor of Department of Computer Science and Engineering, Shanghai Jiao Tong University, China. His current research interests include mobile crowd sensing, communication coding algorithms, and wireless communications.
\end{IEEEbiography}


\vspace{-33pt}
\begin{IEEEbiography}[{\includegraphics[width=1in,height=1.25in,clip,keepaspectratio]{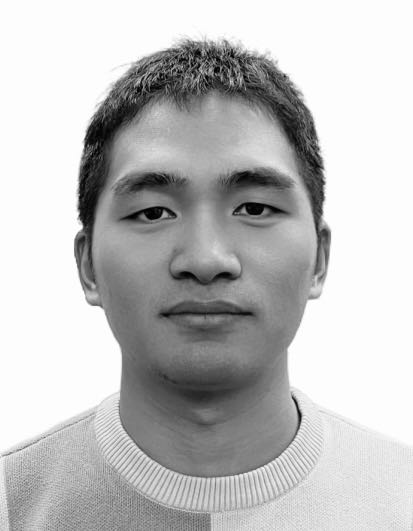}}]{Lei Li} is PhD candidate of Department of Computer Science, University of Copenhagen, Copenhagen, Denmark from 2020. Previously, after he received the master degree of University of Science of Thehology of China, he acted as a researcher of SenseTime for deep learning of image analysis. Now he is working on deep learning for 3D point cloud and multi-model analysis. His research interests include machine learning, computer visionand image processing.
\end{IEEEbiography}

\vspace{-33pt}
\begin{IEEEbiography}[{\includegraphics[width=1in,height=1.25in,clip,keepaspectratio]{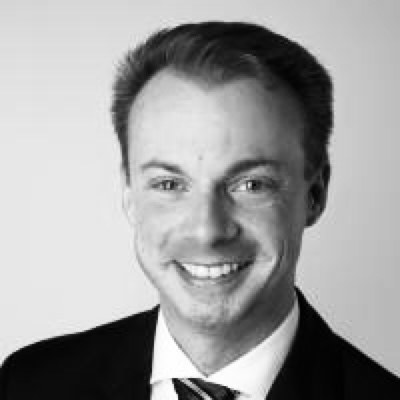}}]{Boris D$\ddot{u}$dder} is an associate professor at the department of computer science (DIKU) at the University of Copenhagen (UCPH), Denmark. He is head of the research group Software Engineering \& Formal Methods at DIKU. His primary research interests are formal methods and programming languages in software engineering of trustworthy distributed systems, where he is studying automated program generation for adaptive systems with high-reliability guarantees. He is working on the computational foundations of reliable and secure Big Data ecosystems. His research is bridging the formal foundations of computer science and complex industrial applications. 
\end{IEEEbiography}

\vspace{-33pt}
\begin{IEEEbiography}[{\includegraphics[width=1in,height=1.25in,clip,keepaspectratio]{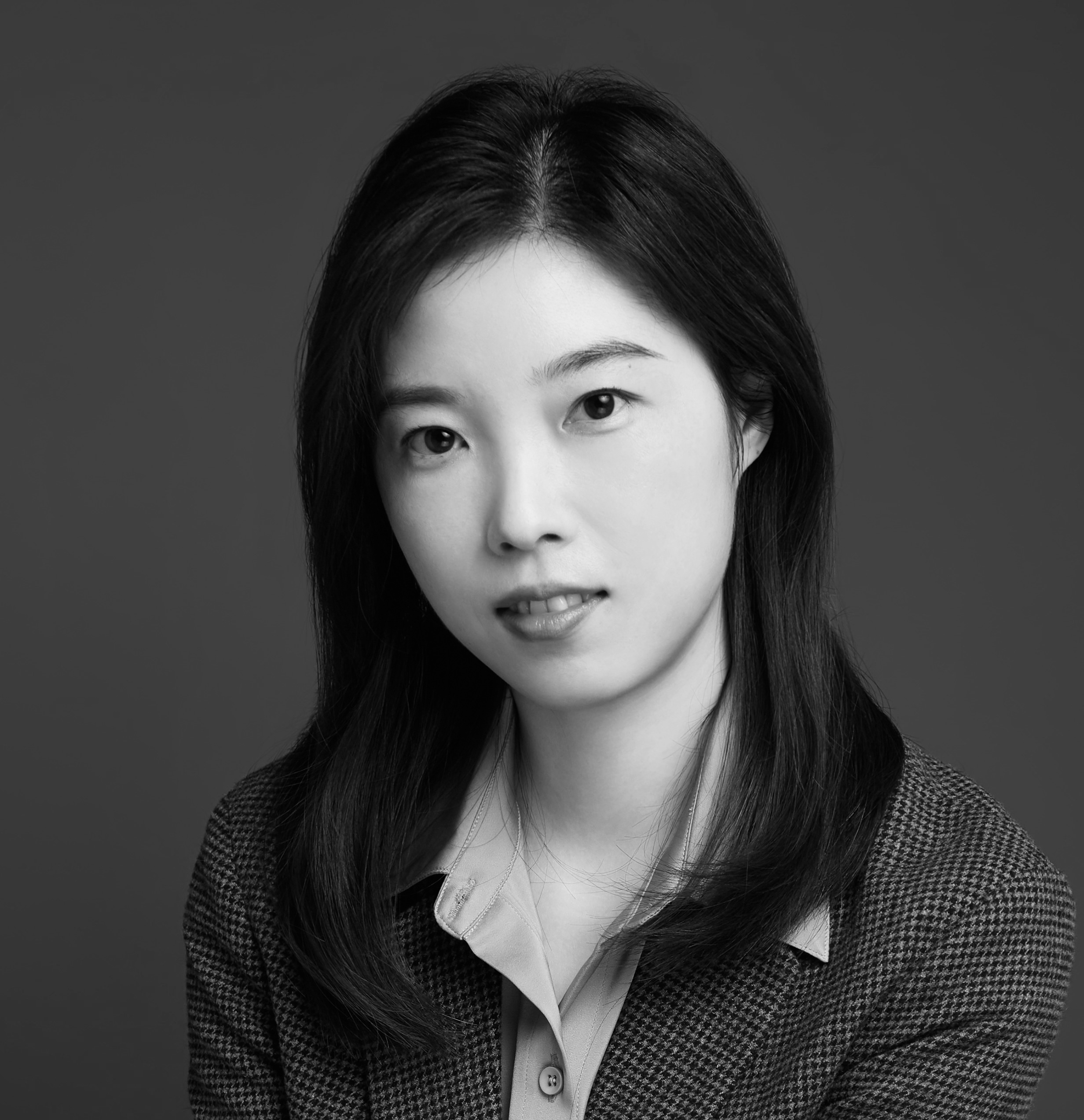}}]{Haiqin Wu} received her B.E. degree in Computer Science and Ph.D. degree in Computer Application Technology from Jiangsu University in 2014 and 2019, respectively. She is an Associate Professor at the Shanghai Key Laboratory of Trustworthy Computing (Software Engineering Institute), East China Normal University, China. Before joining ECNU, she was a postdoctoral researcher in the Department of Computer Science, University of Copenhagen, Denmark. She was also a visiting student in the School of Computing, Informatics, and Decision Systems Engineering at Arizona State University, US. Her research interests include data security and privacy protection, mobile crowdsensing/crowdsourcing, and blockchain-based applications. 
\end{IEEEbiography}

\vspace{-33pt}
\begin{IEEEbiography}[{\includegraphics[width=1in,height=1.25in,clip,keepaspectratio]{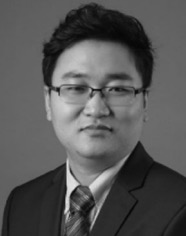}}]{Haiming Jin} received the BS degree from Shanghai Jiao Tong University, Shanghai, China, in 2012, and the PhD degree from the University of Illinois at UrbanaChampaign (UIUC), Urbana, IL, in 2017. He is currently a tenure-track associate professor with the John Hopcroft Center for Computer Science and the Department of Electronic Engineering, Shanghai Jiao Tong University. Before this, he was a post-doctoral research associate with the Coordinated Science Laboratory, UIUC. His research interests include crowd and social sensing systems, reinforcement learning, and mobile pervasive and ubiquitous computing. 
\end{IEEEbiography}


\vspace{-33pt}
\begin{IEEEbiography}[{\includegraphics[width=1in,height=1.25in,clip,keepaspectratio]{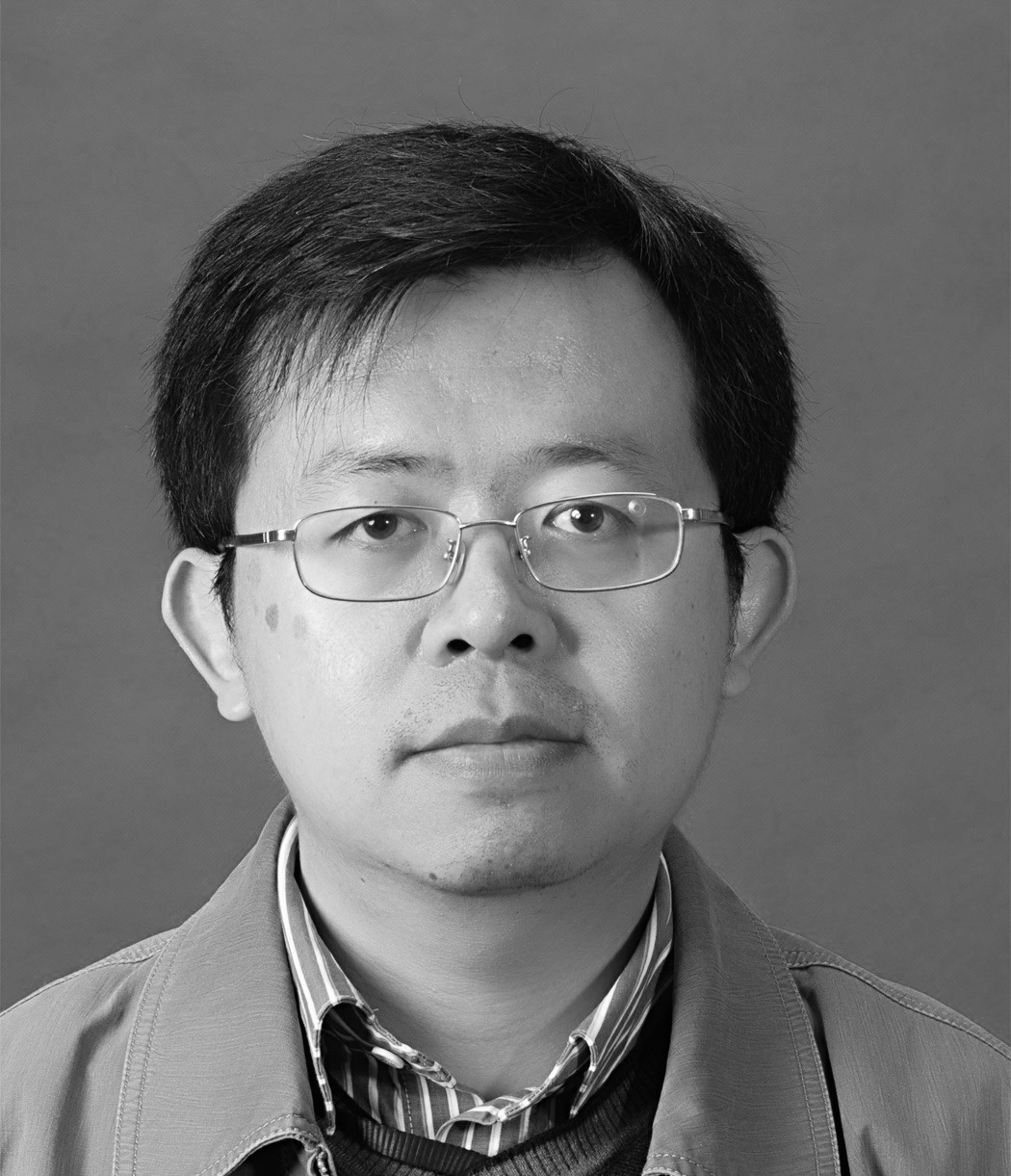}}]{Yuan Luo (Member, IEEE)} received the B.S. degree in applied mathematics, and the M.S. and Ph.D. degrees in probability statistics from Nankai University, Tianjin, China, in 1993, 1996, and 1999, respectively. From July 1999 to April 2001, he was a Post-Doctoral Researcher with the Institute of Systems Science, Chinese Academy of Sciences, Beijing, China. From May 2001 to April 2003, he was a Post-Doctoral Researcher with the Institute for Experimental Mathematics, University of Duisburg–Essen, Essen, Germany. Since June 2003, he has been with the Department of Computer Science and Engineering, Shanghai Jiao Tong University, Shanghai, China. Since 2006, he has been a Full Professor and the Vice Director of the Department, from 2016 to 2018 and since 2021. His current research interests include coding theory, information theory, and big data analysis.
\end{IEEEbiography}


\newpage

\section{appendix}

\subsection{Proof of Theorem 2}

To prove the truthfulness of HERALD*, we will demonstrate its adherence to the criteria outlined in Theorem~\ref{truthful}.

\begin{theorem}[\cite{Singer2010focs}]\label{truthful}
A mechanism satisfies truthfulness only if the following conditions are met:
\begin{itemize}
\item[1)]
The selection rule is monotonic: If a matching pair wins by offering a bid of $b_{i}$, it will also win if it bids $b_{i}^{\prime}\leq b_{i}$;
\item[2)]
Each winning pair is compensated with the critical value: A matching pair will not win if it bids higher than this value.
\end{itemize}
\end{theorem}


\begin{proof}

\textbf{Monotonicity:} Given a matching pair $(\Gamma_{j},b_i)$, we will prove that if it wins with a bid of $b_{i}$, it will also win with a bid of $b_{i}^{\prime}\leq b_{i}$. We will demonstrate this in the following two scenarios.

\emph{Case 1:}
During a winning selection phase iteration, if the CF of the winning pair $(\Gamma_{j},b_i)$ satisfies $\frac{b_{i}}{|\Gamma_{j}\cap\mathcal{T}|}\leq \frac{T}{|\mathcal{T}|}$, then it implies that it possesses the smallest CF among all the matching pairs. Consequently, it will also win with a bid of $b_{i}^{\prime}\leq b_{i}$.

\emph{Case 2:}
During an iteration, if the CF of the winning pair $(\Gamma_{j},b_i)$ satisfies $\frac{b_{i}}{|\Gamma_{j}\cap\mathcal{T}|}>\frac{T}{|\mathcal{T}|}$, then it indicates that it possesses the least cost among the matching pairs and that there exists no matching pair $(\Gamma_{J},b_I)$ such that $\frac{b_{I}}{|\Gamma_{J}\cap\mathcal{T}|}\leq \frac{T}{|\mathcal{T}|}$. We must subsequently examine two sub-cases.

\emph{Subcase 2.1:} If the bid $b_{i}^{\prime}\leq b_{i}$ satisfies $\frac{b_{i}^{\prime}}{|\Gamma_{j}\cap\mathcal{T}|}>\frac{T}{|\mathcal{T}|}$, then it will also win with a bid of $b_{i}^{\prime}$, as $b_{i}^{\prime}$ is the lowest and there is no matching pair $(\Gamma_{J},b_I)$ such that $\frac{b_{I}}{|\Gamma_{J}\cap\mathcal{T}|}\leq \frac{T}{|\mathcal{T}|}$.

\emph{Subcase 2.2:} If the bid $b_{i}^{\prime}\leq b_{i}$ satisfies $\frac{b_{i}^{\prime}}{|\Gamma_{j}\cap\mathcal{T}|}\leq \frac{T}{|\mathcal{T}|}$, then it will also win by bidding $b_{i}^{\prime}$, as it is the sole matching pair with a CF that is less than or equal to $\frac{T}{|\mathcal{T}|}$.

\textbf{Critical Value:} If a matching pair $(\Gamma_{j},b_i)$ wins, it follows that worker $i$'s payment is $p_{i}=p_{i}+\max\{b_{i},p_{\mathcal{R}_{i}}\}$, where $p_{\mathcal{R}_{i}}=\sum_{\ell\in\mathcal{R}_{i}}b_{\ell}$. If worker $i$ increases his/her bid to $\widetilde{b}_{i}$ such that $\widetilde{b}_{i}\leq p_{\mathcal{R}_{i}}$, his/her payments remain constant. However, if $\widetilde{b}_{i}> p_{\mathcal{R}_{i}}$, we must examine the following two cases during each iteration of the winning selection phase.

\emph{Case 1:} When CF of matching pair $(\Gamma_{j},\widetilde{b}_{i})$ satisfies $\frac{\widetilde{b}_{i}}{|\Gamma_{j}\cap\mathcal{T}|}\leq \frac{T}{|\mathcal{T}|}$, we will prove that there is a matching pair $(\Gamma_{q},b_k)$ where $k\in\mathcal{R}_{i}$ such that $\frac{b_{k}}{|\Gamma_{q}\cap\mathcal{T}|}\leq \frac{\widetilde{b}_{i}}{|\Gamma_{j}\cap\mathcal{T}|}$. We have $\frac{\widetilde{b}_{i}}{|\Gamma_{j}\cap\mathcal{T}|}\geq \frac{\sum_{\ell\in\mathcal{R}_{i}}b_{\ell}}{\sum_{\ell\in\mathcal{R}_{i}}|\Gamma_{\ell}\cap\mathcal{T}|}$. Then let matching pair $(\Gamma_{q},b_k)$ be the one with the minimum CF $\frac{b_{k}}{|\Gamma_{q}\cap\mathcal{T}|}$ among all matching pairs where workers belong to $\mathcal{R}_{i}$, which means that $\frac{b_{k}}{|\Gamma_{q}\cap\mathcal{T}|}\leq \frac{b_{\ell}}{|\Gamma_{\ell}\cap\mathcal{T}|}$ for $\forall \ell\in\mathcal{R}_{i}$, i.e., $b_{k}|\Gamma_{q}\cap\mathcal{T}|\leq b_{\ell}|\Gamma_{q}\cap\mathcal{T}|$. Therefore, we have $b_{k}\sum_{\ell\in\mathcal{R}_{i}}|\Gamma_{\ell}\cap\mathcal{T}|\leq|\Gamma_{q}\cap\mathcal{T}|\sum_{\ell\in\mathcal{R}_{i}}b_{\ell}$, i.e., $\frac{b_{k}}{|\Gamma_{q}\cap\mathcal{T}|}\leq\frac{\sum_{\ell\in\mathcal{R}_{i}}b_{\ell}}{\sum_{\ell\in\mathcal{R}_{i}}|\Gamma_{\ell}\cap\mathcal{T}|}$. Since $\frac{b_{k}}{|\Gamma_{q}\cap\mathcal{T}|}\leq\frac{\widetilde{b}_{i}}{|\Gamma_{j}\cap\mathcal{T}|}$, the platform will select matching pair $(\Gamma_{q},b_k)$ instead of $(\Gamma_{j},\widetilde{b}_{i})$ in this iteration.

\emph{Case 2:} When CF of matching pair $(\Gamma_{j},\widetilde{b}_{i})$ satisfies $\frac{\widetilde{b}_{i}}{|\Gamma_{j}\cap\mathcal{T}|}> \frac{T}{|\mathcal{T}|}$, we need to consider two subcases.

\emph{Subcase 2.1:} Once there exist some matching pairs $(\Gamma_{q},\widetilde{b}_{\ell})$ for $\Gamma_{q}\in \mathcal{Y}$ and  $\ell\in\mathcal{R}_{i}$ such that $\frac{b_{\ell}}{|\Gamma_{q}\cap\mathcal{T}|}\leq \frac{T}{|\mathcal{T}|}$, the platform will select a matching pair $(\Gamma_{q},\widetilde{b}_{\ell})$ among them with the minimum CF instead of the matching pair $(\Gamma_{j},\widetilde{b}_{i})$.

\emph{Subcase 2.2:} Once the CFs of all matching pairs $(\Gamma_{q},\widetilde{b}_{\ell})$ for $\Gamma_{q}\in \mathcal{Y}$ and  $\ell\in\mathcal{R}_{i}$ satisfies $\frac{b_{\ell}}{|\Gamma_{q}\cap\mathcal{T}|}> \frac{T}{|\mathcal{T}|}$, the platform will always find a matching pair $(\Gamma_{q},\widetilde{b}_{\ell})$ with the minimum bid $b_{\ell}$ such that $b_{\ell}\leq p_{\mathcal{R}_{i}}\leq \widetilde{b}_{i}$, which means that the platform will not select the matching pair $(\Gamma_{j},\widetilde{b}_{i})$.

Therefore, the conclusion holds.
\end{proof}

\subsection{Proof of Theorem 3}

\begin{proof}
Consider two input bid profiles $\overrightarrow{b}$ and $\overrightarrow{b}'$ that differ in only one bid. Let $M(\overrightarrow{b})$ and $M(\overrightarrow{b}')$ denote the task-worker matching results by HERALD* with inputs $\overrightarrow{b}$ and $\overrightarrow{b}'$, respectively. We aim to prove that HERALD* achieves differential privacy for an arbitrary sequence of task-worker matching results $\mathcal{I} = \{
(\Gamma_1,b_i), (\Gamma_2,b_j),..., (\Gamma_l,b_t)\}$ of length $l$ for $m$ workers, where each worker can match $k$ task subsets with $0\le k \le l$. To analyze the relative probability of HERALD* for the given bid inputs $\overrightarrow{b}$ and $\overrightarrow{b}'$, we consider:
\begin{equation}
\footnotesize
\begin{aligned}
&\frac{Pr[M(\overrightarrow{b})=\mathcal{I}]}{Pr[M(\overrightarrow{b}')=\mathcal{I}]}\overset{(a)}=\prod_{j=1}^l\frac{\frac{exp(-\frac{\epsilon b_{j}}{2 (b_{max}-b_{min})})}{\sum_{i \in \mathcal{W}}exp(-\frac{\epsilon b_i}{2 (b_{max}-b_{min})})}}{\frac{exp(-\frac{\epsilon b_{j}'}{2 (b_{max}-b_{min})})}{\sum_{i \in \mathcal{W}}exp(-\frac{\epsilon b_i'}{2 (b_{max}-b_{min})})}}\\
&=\prod_{j=1}^l\frac{exp(-\frac{\epsilon b_{j}}{2 (b_{max}-b_{min})})}{exp(-\frac{\epsilon b_{j}'}{2 (b_{max}-b_{min})})} \times \prod_{j=1}^l\frac{\sum_{i \in \mathcal{W}}exp(-\frac{\epsilon b_i'}{2 (b_{max}-b_{min})})}{\sum_{i \in \mathcal{W}}exp(-\frac{\epsilon b_i}{2 (b_{max}-b_{min})})},
\end{aligned}
\end{equation}
where equation (a) is derived by formula (\ref{Normalize LIN score functuon}) and $b_i$, $b_j$ denote the bids of workers matching the task subsets $\Gamma_i$, $\Gamma_j$, respectively. Then, we prove this theorem in two cases. When $b_k > b_k'$, the value of the first product is at most $1$, we have
\begin{equation}
\footnotesize
\begin{aligned}\label{lin-1}
&\frac{Pr[M(\overrightarrow{b})=\mathcal{I}]}{Pr[M(\overrightarrow{b}')=\mathcal{I}]}
\le \prod_{j=1}^l\frac{\sum_{i \in \mathcal{W}}exp(-\frac{\epsilon b_i'}{2 (b_{max}-b_{min})})}{\sum_{i \in \mathcal{W}}exp(-\frac{\epsilon b_i}{2 (b_{max}-b_{min})})}\\
&=\prod_{j=1}^l\frac{\sum_{i \in \mathcal{W}}exp(\frac{\epsilon (b_i-b_i')}{2 (b_{max}-b_{min})})exp(-\frac{\epsilon b_i}{2 (b_{max}-b_{min})})}{\sum_{i \in \mathcal{W}}exp(-\frac{\epsilon b_i}{2 (b_{max}-b_{min})})}\\
&=\prod_{j=1}^l E_{i \in \mathcal{W}}\left[exp(\frac{\epsilon (b_i-b_i')}{2 (b_{max}-b_{min})})\right]\\
&\overset{(a)}\le \prod_{j=1}^l E_{i \in \mathcal{W}}\left[1+(e-1)(\frac{\epsilon (b_i-b_i')}{2 (b_{max}-b_{min})})\right]\\
&\overset{(b)}\le exp\left( \epsilon(e-1)(\frac{ \sum_{j=1}^l E_{i \in \mathcal{W}} (b_i-b_i')}{2 (b_{max}-b_{min})}) \right)\\
&\overset{(c)}\le exp\left( \epsilon(e-1)(\frac{ \sum_{j=1}^l (b_{max}-b_{min})}{2 (b_{max}-b_{min})}) \right) = exp\left(\frac{\epsilon(e-1)l}{2}\right),
\end{aligned}
\end{equation}
where the inequality (a) holds because for all $x\le 1, e^x\le 1+(e-1)x$. The inequality (b) holds because for all $x\in R, 1+x\le e^x$. In Section \ref{subsecsys}, it was mentioned that the bid $b_{i}$ of every worker $i$ is confined within the interval $[b_{min},b_{max}]$, where $b_{min}$ is normalized to $1$ and $b_{max}$ is a fixed constant such that inequality (c) is satisfied.

When $b_k \le b_k'$, the value of the second product is at most $1$, we have
\begin{equation}
\footnotesize
\begin{aligned}\label{lin-2}
&\frac{Pr[M(\overrightarrow{b})=\mathcal{I}]}{Pr[M(\overrightarrow{b}')=\mathcal{I}]}\le \prod_{j=1}^l\frac{exp(-\frac{\epsilon b_{j}}{2 (b_{max}-b_{min})})}{exp(-\frac{\epsilon b_{j}'}{2 (b_{max}-b_{min})})}\\
&=\prod_{j=1}^l exp\left(\frac{\epsilon(b_{j}'-b_{j})}{2 (b_{max}-b_{min})}\right) =exp\left(\frac{\epsilon}{2 (b_{max}-b_{min})}\sum_{j=1}^l(b_{j}'-b_{j})\right)\\
&\overset{(a)}\le exp\left(\frac{\epsilon}{2 (b_{max}-b_{min})} \times l (b_{max}-b_{min}) \right) =exp\left(\frac{\epsilon l}{2}\right),
\end{aligned}
\end{equation}
in which the inequality (a) holds because of the same reason of inequality (c) in formula (\ref{lin-1}). Combining the formulas (\ref{lin-1}) and (\ref{lin-2}), the proof is completed.
\end{proof}

\subsection{Proof of Theorem 4}

\begin{proof}
The proof is analogous to that of the linear score function, and thus we will not repeat the default setting. Instead, we examine the HERALD* relative probability for a particular input bid profile $\overrightarrow{b}$ and its perturbed version $\overrightarrow{b}'$:

\begin{equation}
\footnotesize
\begin{aligned}
&\frac{Pr[M(\overrightarrow{b})=\mathcal{I}]}{Pr[M(\overrightarrow{b}')=\mathcal{I}]}\overset{(a)}=\prod_{j=1}^l\frac{\frac{exp\left(\frac{-\epsilon \ln \frac{b_j}{{b_{max}}} }{2 \ln b_{max}}\right)}{\sum_{i \in \mathcal{W}}exp\left(\frac{-\epsilon \ln \frac{b_i}{{b_{max}}} }{2 \ln b_{max}}\right)}}{\frac{exp\left(\frac{-\epsilon \ln \frac{b_j'}{{b_{max}}} }{2 \ln b_{max}}\right)}{\sum_{i \in \mathcal{W}}exp\left(\frac{-\epsilon \ln \frac{b_i'}{{b_{max}}} }{2 \ln b_{max}}\right)}}\\
&=\prod_{j=1}^l\frac{exp\left(\frac{-\epsilon \ln \frac{b_j}{{b_{max}}} }{2 \ln b_{max}}\right)}{exp\left(\frac{-\epsilon \ln \frac{b_j'}{{b_{max}}} }{2 \ln b_{max}}\right)} \times \prod_{j=1}^l\frac{\sum_{i \in \mathcal{W}}exp\left(\frac{-\epsilon \ln \frac{b_i'}{{b_{max}}} }{2 \ln b_{max}}\right)}{\sum_{i \in \mathcal{W}}exp\left(\frac{-\epsilon \ln \frac{b_i}{{b_{max}}} }{2 \ln b_{max}}\right)},
\end{aligned}
\end{equation}
where equation (a) is derived by formula (\ref{Normalize LN score functuon}). Then, we prove this theorem in two cases. When $b_k > b_k'$, the value of the first product is at most $1$, we have
\begin{equation}
\footnotesize
\begin{aligned}\label{ln-1}
&\frac{Pr[M(\overrightarrow{b})=\mathcal{I}]}{Pr[M(\overrightarrow{b}')=\mathcal{I}]}
\le \prod_{j=1}^l\frac{\sum_{i \in \mathcal{W}}exp\left(\frac{-\epsilon \ln \frac{b_i'}{{b_{max}}} }{2 \ln b_{max}}\right)}{\sum_{i \in \mathcal{W}}exp\left(\frac{-\epsilon \ln \frac{b_i}{{b_{max}}} }{2 \ln b_{max}}\right)}\\
&=\prod_{j=1}^l\frac{\sum_{i \in \mathcal{W}}exp\left(\frac{-\epsilon }{2}\left(\frac{\ln \frac{b_i'}{{b_{max}}}}{\ln b_{max}}-\frac{\ln \frac{b_i}{{b_{max}}}}{\ln b_{max}}\right)\right)exp\left(\frac{-\epsilon \ln \frac{b_i}{{b_{max}}} }{2 \ln b_{max}}\right)}{\sum_{i \in \mathcal{W}}exp\left(\frac{-\epsilon \ln \frac{b_i}{{b_{max}}} }{2 \ln b_{max}}\right)}\\
&=\prod_{j=1}^l E_{i \in \mathcal{W}}\left[exp\left(\frac{-\epsilon \ln \frac{b_i'}{b_i}}{2 \ln b_{max}} \right)\right]\\
&\overset{(a)}\le \prod_{j=1}^l E_{i \in \mathcal{W}}\left[1+(e-1)\left(\frac{-\epsilon \ln \frac{b_i'}{b_i}}{2 \ln b_{max}} \right)\right]\\
&\overset{(b)}\le exp\left(\frac{ \epsilon(e-1)}{2}\left(\frac{\sum_{j=1}^l E_{i \in \mathcal{W}}\ln \frac{b_i}{b_i'}}{\ln b_{max}} \right)\right)\\
&\overset{(c)}\le exp\left(\frac{ \epsilon(e-1)}{2}\left(\frac{\sum_{j=1}^l \ln \frac{b_{max}}{b_{min}}}{\ln b_{max}} \right)\right) =exp\left(\frac{ \epsilon(e-1)l}{2}\right),
\end{aligned}
\end{equation}
where inequality (a) is valid because, for all $x\le 1, e^x\le 1+(e-1)x$. The inequality (b) is valid because, for all $x\in R, 1+x\le e^x$. The inequality (c) holds because of the same reason of inequality (c) in formula (\ref{lin-1}).

When $b_k \le b_k'$, the value of the second product is at most $1$, we have
\begin{equation}
\footnotesize 
\begin{aligned}\label{ln-2}
\frac{Pr[M(\overrightarrow{b})=\mathcal{I}]}{Pr[M(\overrightarrow{b}')=\mathcal{I}]}&\le \prod_{j=1}^l\frac{exp\left(\frac{-\epsilon \ln \frac{b_j}{{b_{max}}} }{2 \ln b_{max}}\right)}{exp\left(\frac{-\epsilon \ln \frac{b_j'}{{b_{max}}} }{2 \ln b_{max}}\right)}\\
&=\prod_{j=1}^l exp\left( \frac{\epsilon}{2\ln b_{max}} \left(\ln \frac{b_{j}'}{{b_{max}}} - \ln \frac{b_{j}}{{b_{max}}} \right)\right) \\
&= \prod_{j=1}^l exp\left( \frac{\epsilon}{2\ln b_{max}} \ln \frac{b_{j}'}{b_{j}}\right)\\
&=exp\left(\frac{\epsilon}{2\ln b_{max}}\sum_{j=1}^l \ln \frac{b_{j}'}{b_{j}}\right)\\
&\overset{(a)}\le exp\left(\frac{\epsilon}{2 \ln b_{max} }\sum_{j=1}^l\left[\ln \frac{b_{max}}{b_{min}}\right]\right) \\
&= exp\left(\frac{\epsilon l}{2}\right),
\end{aligned}
\end{equation}
in which the inequality (a) holds because of the same reason of inequality (c) in formula (\ref{lin-1}). Combining the formulas (\ref{ln-1}) and (\ref{ln-2}), the proof is completed, then the proof is completed.
\end{proof}

\subsection{Proof of Lemma 2}

\begin{proof}
Assume that HERALD* selects workers for the winning set $\mathcal{S}$ using type I selection in the following order: $\mathcal{S}_{I}=\{1,\ldots,h\}$. Let $\widetilde{\mathcal{T}}_{i}$ denote the set of tasks whose sensory data has not been collected just before worker $i$ is selected. Since HERALD* carries out type I selection, $c_{i}\leq|\Gamma_{i}\cap\widetilde{\mathcal{T}}_{i}|\frac{64\mathbb{E}[C_{\mathcal{OPT}}(\mathcal{A},W)]}{|\widetilde{\mathcal{T}}_{i}|}$, where $\mathcal{A}$ is a subset of $k$ tasks possibly arriving simultaneously from $\mathcal{T}$. Hence, the social cost of workers in $\mathcal{S}_{I}$ can be bounded by
\begin{equation}
\footnotesize
\begin{split}
\sum_{i\in\mathcal{S}_{I}}c_{i}&\leq\sum_{i\in\mathcal{S}_{I}}\frac{64|\Gamma_{i}\cap\widetilde{\mathcal{T}}_{i}|\mathbb{E}[C_{\mathcal{OPT}}({\mathcal{A},W})]}{|\widetilde{\mathcal{T}}_{i}|}
\\
&\leq 64\mathbb{E}[C_{\mathcal{OPT}}(\mathcal{A},W)]\sum_{t=1}^{n}\frac{1}{t},
\end{split}
\end{equation}
which is at most $64\mathbb{E}[C_{\mathcal{OPT}}(\mathcal{A},W)]\ln n$, in which $\Gamma_i$ denotes the worker $i$'s matching task subset. Therefore, the conclusion holds due to the property of expectation.
\end{proof}

\subsection{Proof of Lemma 3}

\begin{proof}
Let us recall that the set of workers with the bids in the winning set $\mathcal{S}$ selected by HERALD* through type II selection is denoted by $\mathcal{S}_{II}=\{1,\ldots,\ell\}$. Set $k_{\ell+1}=0$ and $c_{0}=0$ for notational convenience. Let $j$ be $k_{j}\geq8\ln 2n$ but $k_{j+1}<8\ln 2n$. Then, we observe that there are at most $8\ln 2n$ tasks from $\widetilde{\mathcal{T}}_{j+1}$ in expectation. As each of these tasks is executed by a worker whose cost is not greater than the one performing it in $\mathcal{S}^{*}(\mathcal{A},W)$, the cost incurred by workers $j+1,\ldots,\ell$ is bounded by $8\ln 2n\mathbb{E}[C_{\mathcal{OPT}}(\mathcal{A},W)]$. Therefore, the expected cost incurred by using the remaining workers $1,\ldots,j$ satisfies
\begin{equation}
\footnotesize
\begin{split}
&\sum_{i=1}^{j}c_{i}\text{Pr}[\mathcal{A}\cap(\Gamma_{i}\cap\widetilde{\mathcal{T}}_{i})\neq\emptyset]\leq\hspace{-0.03in}\sum_{i=1}^{j}\hspace{-0.03in}c_{i}\mathbb{E}[|\mathcal{A}\cap(\Gamma_{i}\cap\widetilde{\mathcal{T}}_{i})|]\hspace{-0.03in}\\
&\overset{\widetilde{\mathcal{T}}_{i+1}\subseteq\widetilde{\mathcal{T}}_{i}\backslash\Gamma_{i}}{\leq}\hspace{-0.03in}\sum_{i=1}^{j}c_{i}\mathbb{E}[|\mathcal{A}\cap(\widetilde{\mathcal{T}}_{i}\backslash\widetilde{\mathcal{T}}_{i+1})|]\\
&\leq\sum_{i=1}^{j}c_{i}(k_{i}-k_{i+1})\overset{c_{0}=0}{\leq}\sum_{i=1}^{j}k_{i}(c_{i}-c_{i-1})
\\
&\overset{(a)}{\leq} \sum_{i=1}^{j}16\mathbb{E}[|\mathcal{S}^{\prime}(\mathcal{A}_{i},W)|]\ln l(c_{i}-c_{i-1})\\
&=16\ln l\biggl(c_{j}\mathbb{E}[|\mathcal{S}^{\prime}(\mathcal{A}_{j+1},W)|]
 \\
 &+\sum_{i=1}^{j}c_{i}\Big(\mathbb{E}[|\mathcal{S}^{\prime}(\mathcal{A}_{i},W)|]-\mathbb{E}[|\mathcal{S}^{\prime}(\mathcal{A}_{i+1},W)|]\Big)\biggl)\\
&\overset{(b)}{\leq}16\ln l\biggl(\mathbb{E}[C(\mathcal{S}^{\prime}(\mathcal{A}_{j+1},W))] \\
&+\sum_{i=1}^{j}\Big(\mathbb{E}[C(\mathcal{S}^{\prime}(\mathcal{A}_{i},W))]-\mathbb{E}[C(\mathcal{S}^{\prime}(\mathcal{A}_{i+1},W))]\Big)\biggl)\\
&\leq16\ln l\mathbb{E}[C_{\mathcal{OPT}}(\mathcal{A},W)],
\end{split}
\end{equation}
\noindent where inequalities (a) and (b) are based on Lemma 3.5 and Lemma 3.4 in reference \cite{grandoni2013set}, $\Gamma_i$ denotes the matching task subset of worker $i$. It was previously shown that $\sum_{i=j+1}^{\ell}c_{i}\text{Pr}[\mathcal{A}\cap(\Gamma_{i}\cap\widetilde{\mathcal{T}}_{i})\neq\emptyset] \leq 8\ln 2n\mathbb{E}[C_{\mathcal{OPT}}(\mathcal{A},W)]$. Therefore, the expected cost incurred by workers $1,\ldots,\ell$ can be bounded as follows:
\begin{equation}
\footnotesize
\sum_{i=1}^{\ell}c_{i}\text{Pr}[\mathcal{A}\cap(\Gamma_{i}\cap\widetilde{\mathcal{T}}_{i})\neq\emptyset] \leq [8\ln 2n + 16\ln l]\cdot \mathbb{E}[C_{\mathcal{OPT}}(\mathcal{A},W)].
\end{equation}
Then we have 
\begin{equation}
\footnotesize
\frac{\sum_{i=1}^{\ell}c_{i}\text{Pr}[\mathcal{A}\cap(\Gamma_{i}\cap\widetilde{\mathcal{T}}_{i})\neq\emptyset]}{\mathbb{E}[C_{\mathcal{OPT}}(\mathcal{A},W)]} \leq \mathcal{O}(\ln ln). 
\end{equation}
The above is based on one matching set $\mathcal{P}$ among all the matching results $\digamma$, we need to take the expectation over all situations, i.e., 
\begin{equation}
\footnotesize
\begin{split}
&\max_{k\in\{1,\ldots,n\}}\mathbb{E}_{\mathcal{P}\in \digamma}\frac{\mathbb{E}_{\mathcal{A}\subseteq\mathcal{T}}[C(\mathcal{S}(\mathcal{A},W))] }{\mathbb{E}_{\mathcal{A}\subseteq\mathcal{T}}[C_{\mathcal{OPT}}(\mathcal{A},W)]}\\
&=\sum_{j=1}^{m^l} P_j \times \frac{\sum_{i=1}^{\ell}c_{i}\text{Pr}[\mathcal{A}\cap(\Gamma_{i}\cap\widetilde{\mathcal{T}}_{i})\neq\emptyset]}{\mathbb{E}[C_{\mathcal{OPT}}(\mathcal{A},W)]}\\
&\leq \sum_{j=1}^{m^l} P_j \times \mathcal{O}(\ln ln) = \mathcal{O}(\ln ln),
\end{split}
\end{equation}
where $P_j$ is the probability of each task-worker matching result. This proof is completed.
\end{proof}

\subsection{Actual Running Time for HERALD*}

\begin{figure}[!t]
\centering
\begin{minipage}{0.85\linewidth}
\centering
 \includegraphics[width=0.85\linewidth]{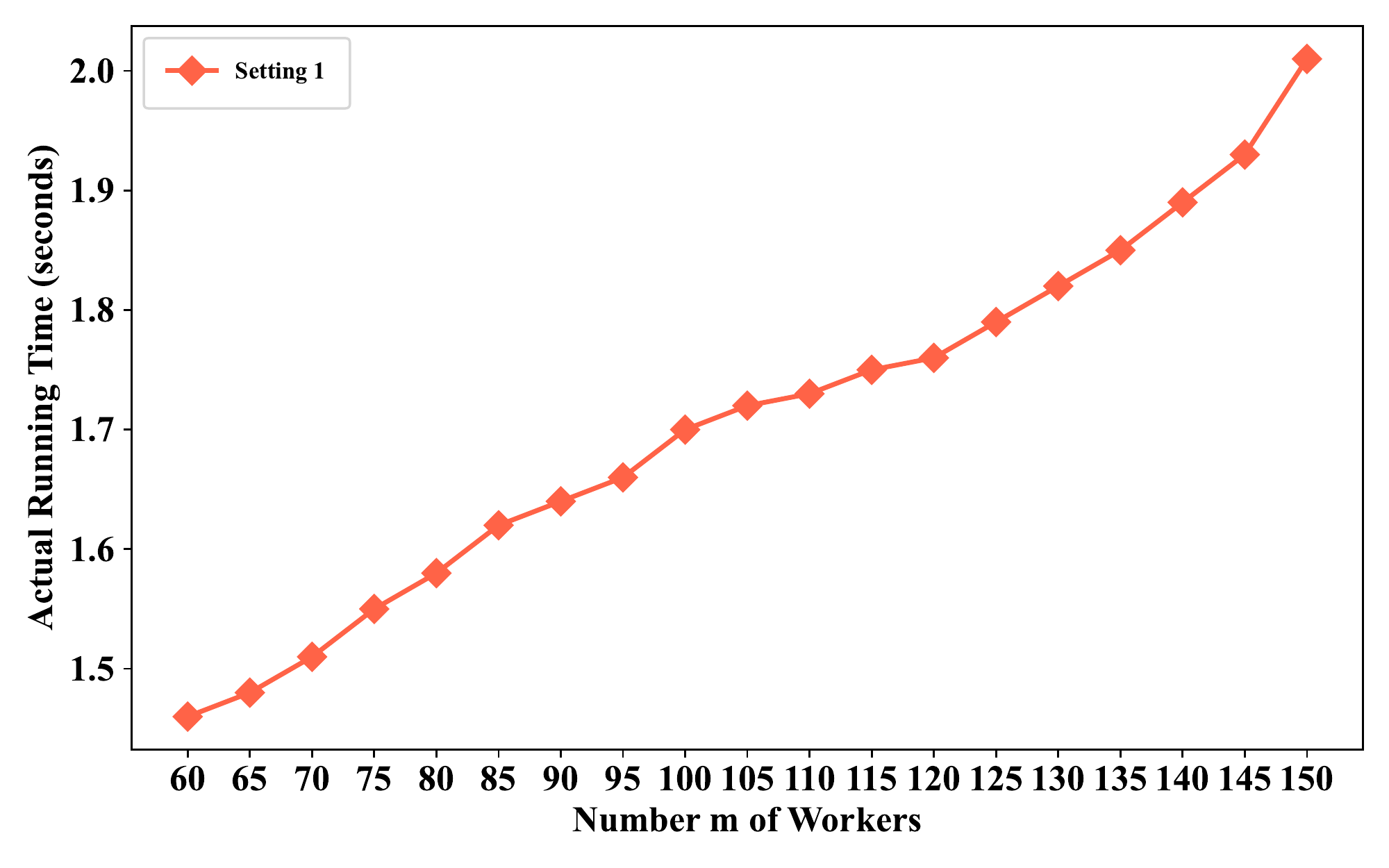}
 \centerline{\footnotesize{\quad (a)}}
 \end{minipage}%
 \qquad
\centering
\begin{minipage}{0.85\linewidth}
\centering
\includegraphics[width=0.85\linewidth]{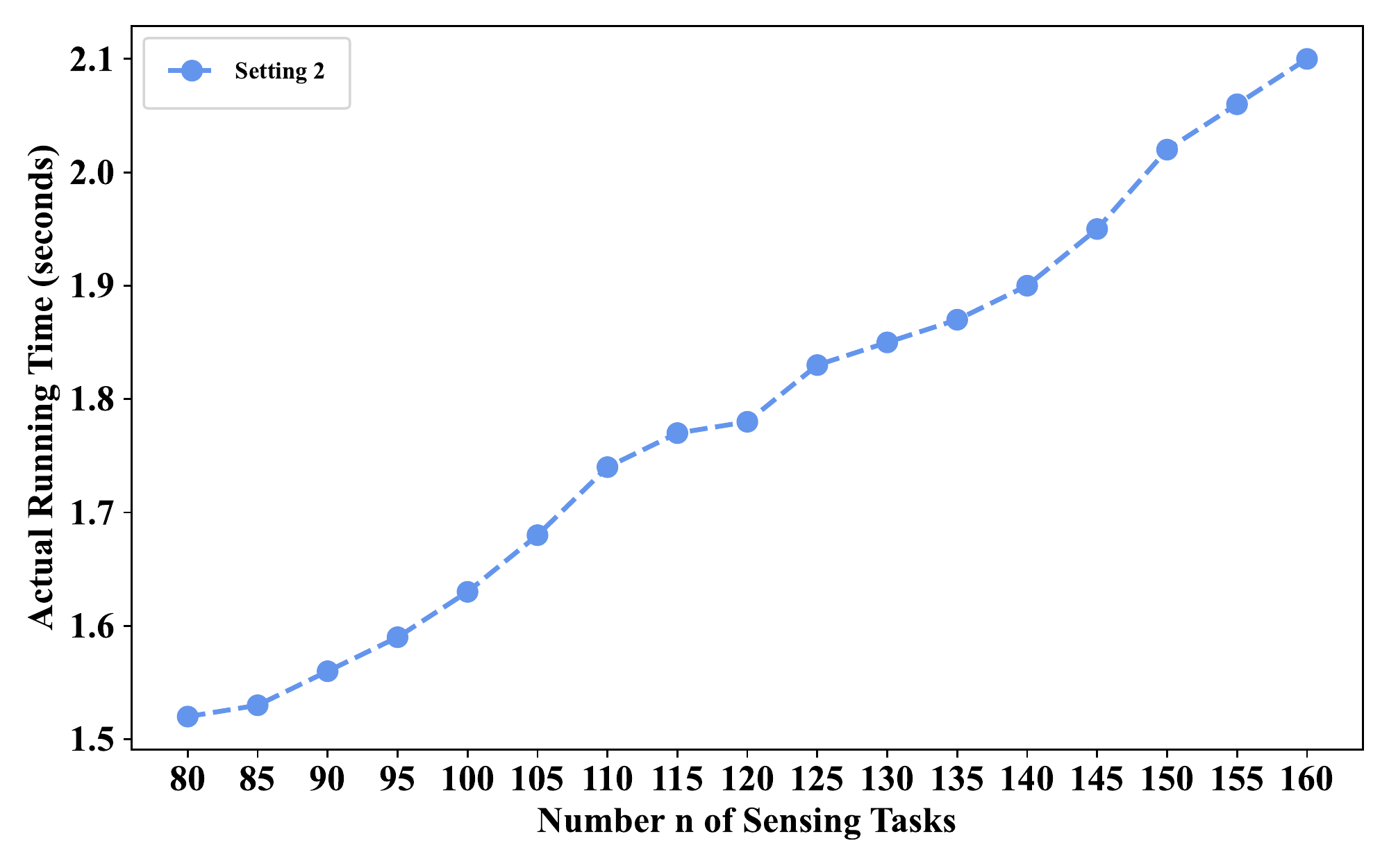}
 \centerline{ \footnotesize{\quad (b)}}
\end{minipage}%
\caption{(a). Actual running time under Setting 1. (b). Actual running time under Setting 2.}\label{Running time}
\end{figure}

To gain a comprehensive understanding of HERALD*'s operational efficiency, we conducted an extensive examination of its actual running times in settings 1 and 2. This process entailed rigorously running the algorithms 100 times across each setting, enabling us to calculate their average running durations. The gathered data, offering a granular view of performance, is vividly represented in Figure \ref{Running time}. As demonstrated in Figure \ref{Running time} (a), there is a clear proportional relationship between actual time and the number of workers. Similarly, Figure \ref{Running time} (b) indicates that the actual time corresponds to the number of sensing tasks. 

What makes these findings particularly noteworthy is their comparison with the theoretical analysis from Proposition \ref{complexity-A}. The empirical results present a significantly more efficient performance of HERALD*, with actual running times being substantially lower than what our theoretical models predicted. This disparity not only underscores the practical effectiveness of HERALD* but also points to potential areas for further optimization and refinement in our theoretical frameworks.


\subsection{Impact of Privacy Parameters $\epsilon$ on HERALD*'s Performance}

Equations (5) and (8) illustrate the probability of worker $i$ with a bid $b_i$ matching a task subset under linear and logarithmic scoring functions. It's evident that at a privacy parameter $\epsilon = 0$, all workers' probabilities follow a uniform distribution. As $\epsilon$ increases, the probability and the expected cost for a worker with bid $b_i$ decrease, thereby reducing the expected social cost and expected total payment. However, this also weakens differential privacy protection. Taking Experimental Setting 1 as an example, we set the privacy parameters to $\epsilon = 0.1$ and $\epsilon = 0.3$ respectively, fix the number of sensing tasks $n$ to $120$, and change the number of workers $m$ from $60$ to $150$ in steps $5$ to evaluate the impact of privacy parameters on the expected social costs and expected total payments of HERALD* under a logarithmic scoring function and that of a linear scoring function. As shown in the figure~\ref{Setting1_epsilon} below, the expected social cost and expected total payment obtained by HERALD* with a larger $\epsilon$ are smaller than those obtained with a smaller $\epsilon$, which is consistent with the above theoretical analysis.

\begin{figure}[ht]
\centering
\begin{minipage}{0.85\linewidth}
\centering
 \includegraphics[width=0.85\linewidth]{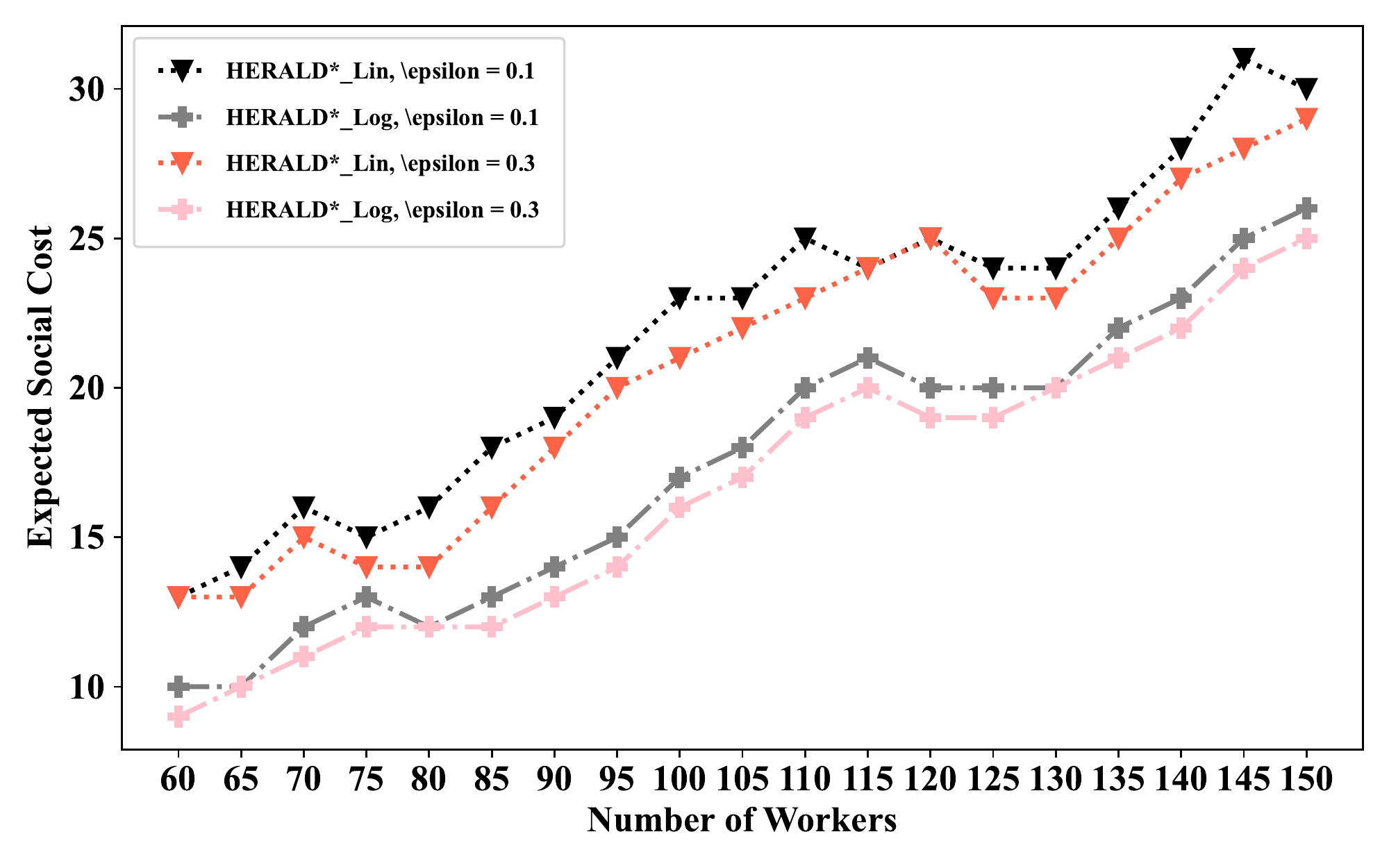}
 \centerline{\footnotesize{\quad (a) }}
 \end{minipage}
 \qquad
\centering
\begin{minipage}{0.85\linewidth}
\centering
\includegraphics[width=0.85\linewidth]{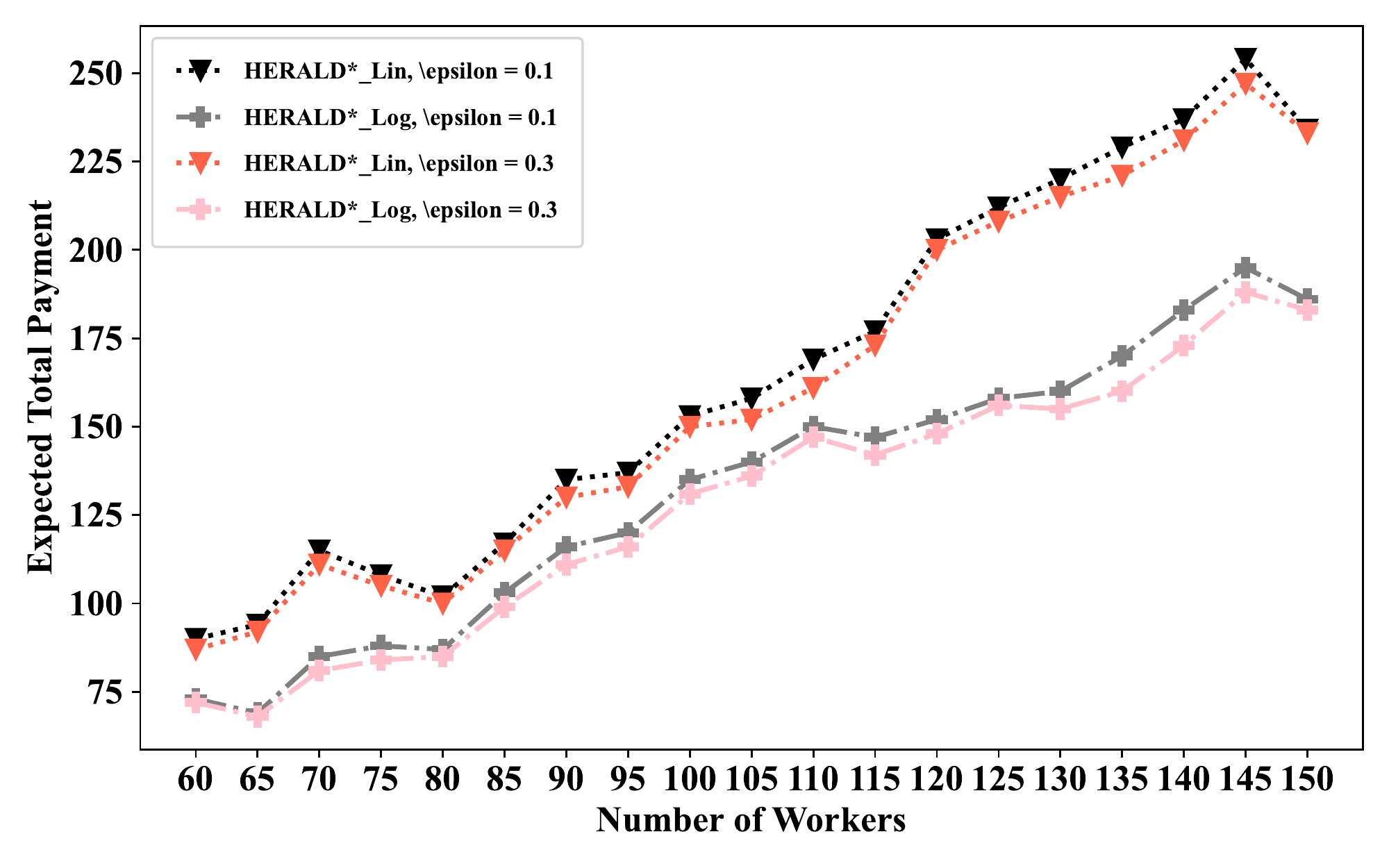}
 \centerline{ \footnotesize{\quad (b)}}
\end{minipage}
\caption{(a). Expected social cost versus different numbers of workers for uncertain tasks. (b). Expected total payment versus different numbers of workers for uncertain tasks.}\label{Setting1_epsilon}
\end{figure}

Interestingly, the expected social cost and expected total payment obtained by HERALD* using the logarithmic score function are lower than those obtained using the linear score function. This is due to the logarithmic score function giving higher chances of selection to users with low bids, resulting in a preference for such users.

\vfill

\end{document}